\documentclass[10pt,journal,compsoc]{IEEEtran}

\usepackage{amsmath, amssymb, mathabx, tikz, subfigure, url, multirow, hhline, ragged2e}
\usepackage[noconfig]{ntheorem}
\usepackage[utf8]{inputenc}
\usetikzlibrary{arrows, automata, decorations.markings}
\newtheorem{theorem}{Theorem}[section]

\newtheorem{corollary}[theorem]{Corollary}
\newtheorem{definition}[theorem]{Definition}
\newtheorem{example}[theorem]{Example}
\newtheorem*{proof}{Proof}

\begin{document}

\title{\huge gMark: Schema-Driven Generation of Graphs and Queries}

\author{Guillaume Bagan, Angela Bonifati, Radu Ciucanu, George H.\,L.\ Fletcher, Aurélien Lemay, Nicky Advokaat

\IEEEcompsocitemizethanks{
\IEEEcompsocthanksitem G.\ Bagan and A.\ Bonifati are with Université Lyon 1, Bat.\ Nautibus, 23-25 av.\ Pierre de Coubertin, 69622 Villeurbanne Cedex, France.\protect\\
{\tt \{guillaume.bagan,angela.bonifati\}@liris.cnrs.fr}
\IEEEcompsocthanksitem R.\ Ciucanu is with Université Blaise Pascal, Clermont-Ferrand, 1 rue de la Chebarde, 63178 Aubière Cedex, France. 
Research partially done at the University of Oxford, supported by EPSRC platform grant DBOnto.\protect\\
{\tt ciucanu@isima.fr}
\IEEEcompsocthanksitem G.\,H.\,L.\ Fletcher and N. Advokaat are with TU Eindhoven, P.O. Box 513, 5600 MB Eindhoven, The Netherlands.\protect\\
{\tt \{g.h.l.fletcher@,n.advokaat@student.\}tue.nl}
\IEEEcompsocthanksitem A.\ Lemay is with Université Lille 3 \& Inria, Parc Scientifique de la Haute Borne, 40 av. Halley, 59650
Villeneuve d'Ascq, France.\protect\\
{\tt aurelien.lemay@inria.fr}
}
}

\IEEEtitleabstractindextext{
\begin{abstract}\justifying
Massive graph data sets are pervasive in contemporary application domains. 
Hence, graph database systems are becoming increasingly important. 
In the experimental study of these systems, it is vital that the research community has shared solutions for the generation of database instances and query workloads having predictable and controllable properties. 
In this paper, we present the design and engineering principles of $\mathsf{gMark}$, a domain- and query language-independent graph instance and query workload generator.
A core contribution of $\mathsf{gMark}$ is its ability to target and control the diversity of properties of both the generated instances and the generated workloads coupled to these instances.
Further novelties include support for regular path queries, a fundamental graph query paradigm, and schema-driven selectivity estimation of queries, a key feature in controlling workload chokepoints. 
We illustrate the flexibility and practical usability of $\mathsf{gMark}$ by showcasing the framework's capabilities in generating high quality graphs and workloads, and its ability to encode user-defined schemas across a variety of application domains.
\end{abstract}

}

\maketitle

\section{Introduction}

{\smallskip\noindent\bf The problem.}
We study the problem of schema-driven generation of synthetic
graph instances and corresponding query workloads for use in experimental
analysis of graph database systems.  Our study is motivated by the ubiquity of
graph data in modern application domains, such as social and biological
networks and geographic databases, to name a few.
In response to these pressures, systems that can handle massive
graph-structured data sets are under intense active research and development.
These systems span from pure graph database systems to more focused knowledge
representation systems.  Native graph databases such as Neo4j~\cite{neo4j} and
Sparksee~\cite{sparksee} propose their own declarative data model and query
language, with particular attention to query optimization, and space and
performance.  In contrast to this trend of specialized systems, general-purpose
systems such as LogicBlox~\cite{ArefCGKOPVW15} rely on declarative solutions
that can cover a broader range of use cases.  Furthermore, knowledge
representation systems such as Virtuoso~\cite{virtuoso} implement the standard
RDF graph data model and SPARQL query language to handle complex navigational
and recursive queries on large-scale Semantic Web data.  

Keeping pace with these developments, synthetic graph and query generation and
benchmarking solutions for graph data management systems has been a
proliferating activity, which started within the Semantic Web community
(e.g.,~\cite{AlucHOD14,SchmidtHLP09}) and recently within the database
community, e.g.,  through the extensive work of the LDBC Council~\cite{Erling2015}.
The latter has been the first to propose a chokepoint-driven design of graph
database benchmarks, which allows fine-grained control of chokepoints of
queries and all aspects of the involved data. 
By relying on a fixed schema and a carefully designed set of benchmark queries,
this lets the community focus on
crucial features of query optimization and/or parallel processing, and decreases
the confusion and the incomparable results of evaluated systems. 

{\smallskip\noindent\bf Our solution.}
We propose a complementary approach in which the focus is not on the design of
individual queries but rather on whole query workloads.  This approach relies
on the control of diversity of both graph schemas and query workloads, which
lets us vary the structural properties of data as well as tailor the generated
queries to a particular domain or application.  We emphasize that a
workload-centric approach primarily targets different benchmarking and
experimental scenarios from the query-centric approach of current benchmarks.
Indeed, we believe that this approach is important 
in contexts in which multiple queries (i.e., those belonging to a query
workload) need to be considered altogether, which occurs in many cases, such as in multi-query
optimization, workload-driven database tuning, streaming applications,
mapping discovery, and query rewriting in data integration systems. 

In this paper, we present $\mathsf{gMark}$, a fully-functional practical 
framework that realizes this workload-centric perspective.  $\mathsf{gMark}$ takes a
schema-driven approach to the flexible and tightly-controlled generation of
synthetic graph instances coupled with sophisticated query workloads.  
The complete framework is provided as {\em open-source
software}\footnote{\url{https://github.com/graphMark/gmark}} ready for use in the
graph processing community.

We further note that the core of query workload generation in $\mathsf{gMark}$ revolves around a
novel method we propose here for selectivity estimation for graph queries,
which is, to the best of our knowledge, the first of its kind.  This method is
of independent interest, applicable to other contexts such as query
optimization and query inference on graphs and, more generally, relational data
and queries. 

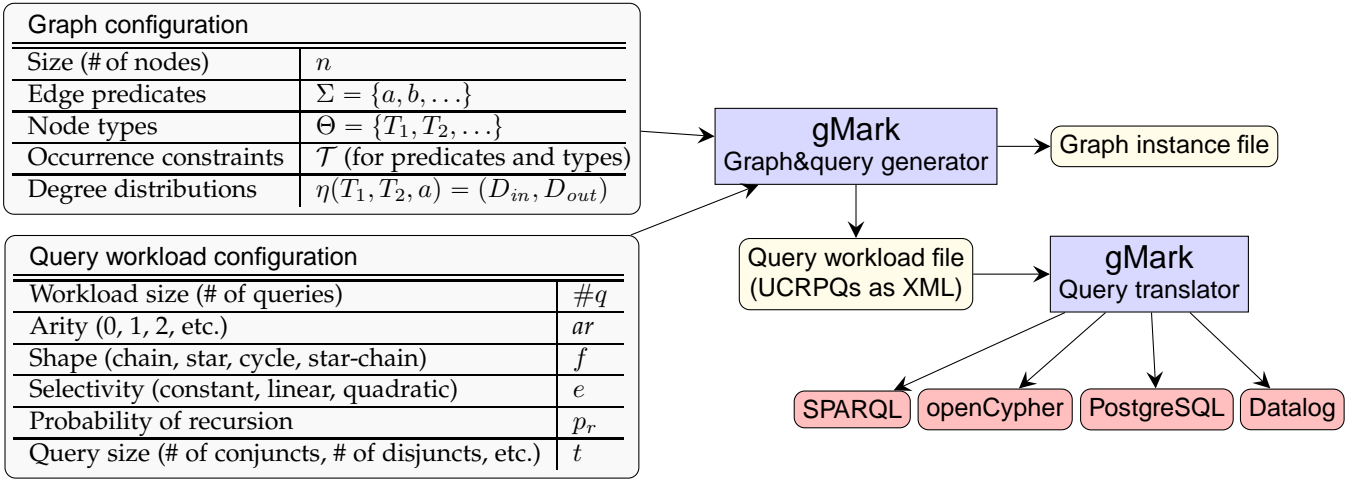
\begin{figure*}
  \centering
  \begin{tikzpicture}[>=stealth,every node/.style={shape=rectangle,draw,rounded corners,align=center},]

\node (gc) [xshift=-0.9cm,fill=gray!5] {
\begin{tabular}{l|l}
\multicolumn{2}{l}{\sf Graph configuration} \\\hline\hline
Size (\#\,of nodes) & $n$\\\hline
Edge predicates & $\Sigma = \{a,b,\ldots\}$\\\hline
Node types & $\Theta = \{T_1,T_2,\ldots\}$\\\hline
Occurrence constraints &  $\mathcal{T}$ (for predicates and types)\hspace*{-0.4cm}\\\hline
Degree distributions & $\eta(T_1,T_2,a)=(D_{\mathit{in}}, D_{\mathit{out}})$
\end{tabular}
};

\node (qc) [xshift=-0.9cm,yshift=-3.3cm,fill=gray!5] {
\begin{tabular}{l|l}
\multicolumn{2}{l}{\sf Query workload configuration} \\\hline\hline
Workload size (\# of queries) & $\#q$\\\hline
Arity (0, 1, 2, etc.) & {\em ar}\\\hline
Shape (chain, star, cycle, star-chain) & $f$\\\hline
Selectivity (constant, linear, quadratic) & $e$\\\hline
Probability of recursion & $p_r$\\\hline
Query size (\# of conjuncts, \# of disjuncts, etc.) & $t$\\
\end{tabular}
};

\node (c2) [xshift=6.2cm, yshift=-0.5cm, rounded corners=0cm,fill=blue!15]{\large \textsf{gMark}\\ \textsf{Graph\&query generator}};

\node (c3) [xshift=10.3cm, yshift=-0.5cm, fill=yellow!10]{\textsf{Graph instance file}} ;

\node (c4) [xshift=6.2cm, yshift=-2.2cm,fill=yellow!10]{\textsf{Query workload file} \\ \textsf{(UCRPQs as XML)}};

\node (c5) [xshift=10.1cm, yshift=-2.2cm, rounded corners=0cm,fill=blue!15]{\large \textsf{gMark} \\ \textsf{Query translator}};

\node (c6) [xshift=6.15cm,yshift=-4cm,fill=red!25] {\sf SPARQL};
\node (c7) [xshift=8.05cm,yshift=-4cm,fill=red!25] {\sf openCypher};
\node (c8) [xshift=10.2cm,yshift=-4cm,fill=red!25] {\sf PostgreSQL};
\node (c9) [xshift=12cm,yshift=-4cm,fill=red!25] {\sf Datalog};

	\draw[decoration={markings,mark=at position 1 with {\arrow[scale=2]{>}}}, postaction={decorate}, shorten >=0.4pt] (gc) to (c2);
	\draw[decoration={markings,mark=at position 1 with {\arrow[scale=2]{>}}}, postaction={decorate}, shorten >=0.4pt] (qc) to (c2);
	\draw[decoration={markings,mark=at position 1 with {\arrow[scale=2]{>}}}, postaction={decorate}, shorten >=0.4pt] (c2) -- (c4);
	\draw[decoration={markings,mark=at position 1 with {\arrow[scale=2]{>}}}, postaction={decorate}, shorten >=0.4pt] (c2) -- (c3);
	\draw[decoration={markings,mark=at position 1 with {\arrow[scale=2]{>}}}, postaction={decorate}, shorten >=0.4pt] (c4) -- (c5);
	\draw[decoration={markings,mark=at position 1 with {\arrow[scale=2]{>}}}, postaction={decorate}, shorten >=0.4pt] (c5) -- (c6);
	\draw[decoration={markings,mark=at position 1 with {\arrow[scale=2]{>}}}, postaction={decorate}, shorten >=0.4pt] (c5) -- (c7);
	\draw[decoration={markings,mark=at position 1 with {\arrow[scale=2]{>}}}, postaction={decorate}, shorten >=0.4pt] (c5) -- (c8);
	\draw[decoration={markings,mark=at position 1 with {\arrow[scale=2]{>}}}, postaction={decorate}, shorten >=0.4pt] (c5) -- (c9);
  \end{tikzpicture}
  \caption{Overview of the $\mathsf{gMark}$ workflow.}
  \label{fig:gmark_arch}
\end{figure*}

\subsection{gMark design principles} 
To support a broad range of systems and domains, we aim in the $\mathsf{gMark}$ design to cover
features and capabilities commonly found in graph query processing, graph
analytics, and schema validation. Reflecting this, the interplay between the
following key features characterizes the $\mathsf{gMark}$ system architecture. To the best of our
knowledge, $\mathsf{gMark}$ is the first framework to satisfy all of these criteria.

{\em Built-in support for schema definition.} A major goal of $\mathsf{gMark}$ is
to account for an array of fundamental user-defined schema constraints during
graph generation.  
In general, $\mathsf{gMark}$ is domain-independent, while it can be used to target a rich variety of realistic domains. 
Hence, graph instance generation leverages an optional schema
definition, called a {\em graph configuration}, which includes the enumeration
of predicates (i.e., edge labels) and node types (i.e., node labels) 
occurring in the data, along with their properties in generated instances (see Fig.~\ref{fig:gmark_arch}).

{\em Controlled instance and workload diversity.} Given a graph configuration, a
subsequent challenge is to exploit it for query workload generation.
Current approaches in query generation rely on the graph instances
to generate queries with desired behavior.
However, this approach is unfeasible for large and loosely
structured networks. We argue that query workload generation must primarily
rely on the graph configuration rather than on the generated graph instances,
while still enforcing the desired behavior of the
generated queries.  Towards this, $\mathsf{gMark}$ supports a broad range of schema-driven parameters in both the graph
and query workload generator (see Fig.~\ref{fig:gmark_arch}), e.g., capable 
of gauging both navigational and recursive query execution performance,
while not relying on a fixed set of query templates. 
To the best of our knowledge, existing solutions for graph
databases do not support such configurable (and easily extensible) range of parameters.

{\em Language- and system-independence.} Another
important design principle of $\mathsf{gMark}$ is 
independence from particular query language syntaxes or systems.
Towards this, $\mathsf{gMark}$ supports various practical output formats for the graphs and for the queries, including N-triples for data,
and SQL, SPARQL, Datalog and openCypher as concrete query language syntaxes for query workloads (see Fig.~\ref{fig:gmark_arch}).  $\mathsf{gMark}$ is also
easily extensible to support other output formats. 

{\em Broad applicability.} Finally, $\mathsf{gMark}$ aims to broadly support a
wide range of application domains.  As an example, we show below 
that it is easy to adapt the scenarios (modulo incomparable features) of
three existing state-of-the-art benchmarks into meaningful $\mathsf{gMark}$ configurations, while also adding new $\mathsf{gMark}$ features: the LDBC Social
Network Benchmark~\cite{Erling2015}, SP2Bench~\cite{SchmidtHLP09}, and the
Waterloo SPARQL Diversity Test Suite (WatDiv)~\cite{AlucHOD14}.

\subsection{Contributions and organization}\label{sec:contributions:organization}
We depict an overview of the $\mathsf{gMark}$ workflow in Fig.~\ref{fig:gmark_arch}. 
The goal of this paper is to present the design, engineering, and first empirical study of the $\mathsf{gMark}$ framework. 
We next outline our main contributions.

$\bullet$ We formalize the problems of graph generation and query workload generation, and show that they are intractable in general (Section~\ref{sec:generation-problem}).

$\bullet$ We provide an in-depth presentation of the $\mathsf{gMark}$ design principles, for the generation of both graphs (Section~\ref{sec:graphgen}) and query workloads (Section~\ref{sec:query-generation-selectivity}).
The most notable novel features are support for recursive queries and query selectivity estimation in the generated query workloads.

$\bullet$ We empirically show the capability of $\mathsf{gMark}$ to cover diverse graphs and query workloads, the accuracy of the estimated selectivities, and the scalability of the generator (Section~\ref{sec:gmarkStudy}).

$\bullet$ We present an in-depth experimental comparison of a representative selection of state-of-the-art graph query engines using $\mathsf{gMark}$, which brings to light important limitations of current graph query processing engines, in particular w.r.t.\ recursive query processing (Section~\ref{sec:databasesStudy}).

We close this section by noting that $\mathsf{gMark}$ supports the full range of data and query features, and practical query syntaxes discussed above. 
 Query selectivity tuning is supported only on binary queries (i.e., queries of arity two).  
 This should not be considered as a limitation since already selectivity tuning is a non-trivial problem for binary queries, and such queries already make interesting major practical cases of query benchmarking in graph databases. For example, all regular path queries (which appear as ``property paths'' in SPARQL~1.1) are binary.

\section{Related work}\label{sec:relatedwork}

Data generation and benchmarking frameworks have played an important role in database systems
research over the last decades,  where efforts such as the TPC Benchmarks and
XML benchmarking suites have been crucial for advancing the state of the art
\cite{BarbosaMY09b,Gray93}.  
Similarly, in support of the experimental study of graph data management solutions, a variety of
synthetic graph tools such as SP2Bench~\cite{SchmidtHLP09}, LDBC~\cite{Erling2015}, LUBM~\cite{GuPaHe05}, BSBM~\cite{BizerS09}, Grr~\cite{BlumC11} and WatDiv~\cite{AlucHOD14} have been
developed in the research community.  
Complementary to the $\mathsf{gMark}$ approach,  
application-driven derivation of graph configurations and synthetic graph instances 
has also recently been studied \cite{Apples11,QiaoO15,ZhangT16}.
Furthermore, extensive
collections of real-world graphs such as
SNAP~\cite{snap} and KONECT~\cite{konect}
are now available as community resources.

Currently available resources 
(i) rely on fixed graphs or
instances of fixed
graph schemas, or (ii) provide limited or no support for generating tailored
query workloads to accompany graph instances.  
These aspects are difficult to jointly relax,
especially in the context of loosely structured complex
networks.  Indeed, there is no community consensus on schema formalisms for
graph data, an area that is still in an early stage of investigation~\cite{BonevaGHPSS14}.  Furthermore, constructing workloads with given selectivity
and structural features is a very difficult problem \cite{ArasuKL11,BrunoCT06,gb2014,LoBKOH10,LoCLHC14,MishraKZ08,Poess04}.  
As mentioned in the Introduction, 
current approaches such as WatDiv and LDBC perform selectivity
estimation on generated graph instances, which becomes unfeasible when
dealing with massive graphs and query workloads.

In $\mathsf{gMark}$, we address this challenge by generating tailored workloads
directly from the graph schema, where 
query selectivity is set as one of the input parameters.
While BSBM and WatDiv do support a workload-centric approach, they do not provide this fine-grained level of control of query behavior.
In general, we are not aware of any 
solutions for controlling selectivity during query generation relying solely on graph schemas.

Finally, to our knowledge $\mathsf{gMark}$ is the first solution for generating workloads exhibiting recursive path queries.
In particular, the queries generated by $\mathsf{gMark}$ are the so-called {\em unions of conjunctive regular path queries}~\cite{Baeza13}.
 This fundamental
query language covers many graph queries that appear in practice. In particular,
SPARQL~1.1 and openCypher have conjunctive regular path queries as their core
  constructs. They are also expressible in modern
  Datalog-like query languages~\cite{ArefCGKOPVW15} and in SQL:1999. 
  As discussed in the Introduction,
  $\mathsf{gMark}$ supports the output of query workloads in all these concrete query language syntaxes.

\section{The generation problem}\label{sec:generation-problem}
We start this section by intuitively introducing the generation problem via a real-world motivating example that also emphasizes some of the limitations of existing benchmarks (Section~\ref{sec:motivating-example}).
Then, we formalize the benchmark generation problem i.e., generating a graph instance and a query workload on this instance according to
a given set of constraints.  More precisely, we formally define the problems of
{\em graph generation} (Section~\ref{sec:graph-generation}) and {\em query
workload generation} (Section~\ref{sec:query-generation}), and we detail constraint parameters
for both problems.  We conclude the section by showing that the generation problem
is {\em intractable} (Section~\ref{sec:intractability-generation}).

\subsection{Motivating example}\label{sec:motivating-example}
Assume that a user wants to perform an extensive empirical evaluation of a new graph query processing algorithm that she designed.
For this purpose, the user needs to efficiently generate: (i) graphs of different characteristics and sizes (to test the robustness and scalability of her algorithm), and (ii) query workloads sufficiently diverse to highlight strong or weak points of her new development.
Additionally, our user would like to specify all parameters in a declarative way and to be able to simulate real-world scenarios.

For instance, the user would like to generate graphs simulating a bibliographical database that uses a simple schema consisting of 5 node types and 4 edge predicates.
Intuitively, the database consists of {\em researchers} who {\em author} {\em papers} that are {\em published} in {\em conferences} (held in {\em cities}) and that can be {\em extended} to {\em journals}.
Moreover, the user would like to specify constraints on the number of occurrences for both the node types and edge predicates, either as proportions of the total size of the graph or as fixed numbers e.g., as in Fig.~\ref{fig:biblio:constraints:nodes} and~\ref{fig:biblio:constraints:edges}.
For instance, for graphs of arbitrary size, half of the nodes should be authors, but a fixed number of nodes should be cities where conferences are held (in a realistic scenario the number of authors increases over time, whereas the number of cities remains more or less constant).

Moreover, our user wants to specify real-world relationships between types and predicates via schema constraints e.g., as in Fig.~\ref{fig:biblio:constraints:schema}.
For instance, the first line encodes that the number of authors on papers follows a Gaussian distribution (the {\em in-distribution} of the schema constraint), whereas the number of papers authored by a researcher follows a Zipfian (power-law) distribution (the {\em out-distribution} of the schema constraint).
The following lines in Fig.~\ref{fig:biblio:constraints:schema} encode constraints such as: a paper is published in exactly one conference, a paper can be extended or not to a journal, a conference is held in exactly one city, the number of conferences per city follows a Zipfian distribution, etc.

Whereas specifying all aforementioned constraints as an input $\mathsf{gMark}$ {\em graph configuration} (cf.\ Fig.~\ref{fig:gmark_arch}) can be easily done via a few lines of XML, to the best of our knowledge there is no benchmark where they can be specified.
For instance, in SP2Bench~\cite{SchmidtHLP09} (which is also based on a similar bibliographical scenario), all constraints are hardcoded and the only parameter that a user can specify is the size of the graph, which makes it impossible for the user to finely tune schema-related characteristics of the graph.
Moreover, in WatDiv~\cite{AlucHOD14}, although the user can specify similar global constraints on the node types and the out-distributions, the absence of global constraints on the edge predicates and the absence of in-distributions entail important limitations, such as the absence of control on the selectivities of the queries of the generated query workloads.

In $\mathsf{gMark}$ we allow the user to finely tune the selectivities of the generated queries.
For instance, the user can specify that she wants queries that, for any graph size, have constant, linear, or quadratic selectivity (we formally define these selectivity classes later on in the paper).
To the best of our knowledge, no graph database benchmark supports such a feature.
In particular, SP2Bench uses a fixed set of queries, while WatDiv can generate synthetic workloads, but without schema-driven selectivity control.
As another remarkable difference to the state-of-the-art benchmarks, none of them supports recursive queries such as
$(${\tt authors}$\cdot${\tt authors}$^-)^*$ 
which selects all pairs of researchers linked by a co-authorship path
(by $^-$ we denote the predicate inverse and by $^*$ the transitive closure).
As shown in the input $\mathsf{gMark}$ {\em query workload configuration} in Fig.~\ref{fig:gmark_arch}, the user can finely-tune e.g., the structure, size, selectivity of such queries.

\begin{figure}\scriptsize\centering
\subfigure[\label{fig:biblio:constraints:nodes}Node types.]{
\begin{tabular}{|l|l|}\hline
\em Node type & \em Constr. \\\hline\hline
{\tt researcher} & 50\%\\\hline
{\tt paper} & 30\%\\\hline
{\tt journal} & 10\%\\\hline
{\tt conference} & 10\%\\\hline
{\tt city} & 100 (fixed)\\\hline
\end{tabular}
}~~~
\subfigure[\label{fig:biblio:constraints:edges}Edge predicates.]{
\begin{tabular}{|l|l|}\hline
\em Edge predicate & \em Constr. \\\hline\hline
{\tt authors} & 50\%\\\hline
{\tt publishedIn} & 30\%\\\hline
{\tt heldIn} & 10\%\\\hline
{\tt extendedTo} & 10\%\\\hline
\end{tabular}
}

\subfigure[\label{fig:biblio:constraints:schema}In- and out-degree distributions.]{
\begin{tabular}{|l|l|l|}\hline
{\em source\,type} $\underrightarrow{\textit{predicate}}$ {\em target\,type}
& {\em In-distr.} & {\em Out-distr.} \\\hline\hline

{\tt researcher} $\underrightarrow{\texttt{authors}}$ {\tt paper} & Gaussian & Zipfian \\\hline
{\tt paper} $\underrightarrow{\texttt{publishedIn}}$ {\tt conference} & Gaussian & Unif.\ [1,1] \\\hline
{\tt paper} $\underrightarrow{\texttt{extendedTo}}$ {\tt journal} & Gaussian & Unif.\ [0,1] \\\hline
{\tt conference} $\underrightarrow{\texttt{heldIn}}$ {\tt city} & Zipfian & Unif.\ [1,1] \\\hline
\end{tabular}
}
\vspace*{-0.2cm}
\caption{The bibliographical motivating example.}
\vspace*{-0.3cm}
\end{figure}

\subsection{Graph generation}\label{sec:graph-generation}

$\mathsf{gMark}$ generates directed edge-labeled graphs and outputs them in formats that are compatible with the supported query languages.
In this section, we formally define a {\em graph configuration} (cf.\ Fig.~\ref{fig:gmark_arch}), which is essentially a set of constraints that generated graph instances should satisfy.  We start by
giving a definition of the {\em schema} constraints which are the backbone of graph
configurations. 

\begin{definition}\label{def:schema}

A \emph{graph schema} is a tuple $\mathcal{S} = (\Sigma,\Theta,\mathcal{T},\eta)$ where $\Sigma$ is a
finite alphabet of {\em predicates}, $\Theta$ is a finite set of {\em types} such
that each node of the generated graph is associated with exactly one type, $\mathcal{T}$
is a set of constraints on $\Sigma$ and $\Theta$ associating to each
predicate and type either a proportion of its occurrences or a fixed constant
value, and $\eta$ is a partial function associating to a triple consisting of
a pair of input and output types $T_1, T_2$ in $\Theta$ and a symbol $a$ in $\Sigma$, a pair
($D_{in},D_{out}$) of in- and out-degree distributions. 
\end{definition}
Predicates correspond to {\em edge labels}, and in the remainder we use the two terms interchangeably. 
A degree distribution is a {\em probability distribution}, among which $\mathsf{gMark}$ supports {\em uniform}, {\em Gaussian} (also known as normal), and {\em Zipfian} distributions.
For each distribution, the user can specify the relevant parameters (i.e., min and max for uniform, $\mu$ and $\sigma$ for Gaussian, and $s$ for Zipfian).
If the user wants to specify only the in- or the out-distribution, she can mark the other one as {\em nonspecified}.
Notice that the parameters for the in- and out-degree distributions of each
triple $T_1, T_2, a$ have to be consistent in order to guarantee the
compatibility of the number of generated ingoing and outgoing edges. We
discuss the details of this consistency check in Section~\ref{sec:graphgen}. 

\begin{definition} A {\em graph configuration} is a tuple $\mathbb{G} = (n, \mathcal{S})$, where $n$ is the number of nodes of the graph and $\mathcal{S}$ is the schema of the graph. 
\end{definition}

\begin{example}\label{ex:schema:graph}\normalfont
Take a graph configuration $\mathbb{G} = (n, \mathcal{S})$ s.\,t.:

$\bullet$ The graph should have $n=5$ nodes.

$\bullet$ The graph should satisfy the schema $\mathcal{S}=(\Sigma,\Theta,\mathcal{T},\eta)$, where $\Sigma =\{a,b\}$, $\Theta=\{T_1,T_2,T_3\}$, $\mathcal{T}$ is defined as $\mathcal{T}(T_1) = 60\%$, 
$\mathcal{T}(T_2) = 20\%$ and $\mathcal{T}(T_3) = 1$ and $\eta$ is defined as follows (we report only some constraints):
\begin{gather*}
\eta(T_1,T_1,a)=(g,z),~~~~\eta(T_1,T_2,b)=(u,g),\\
\eta(T_2,T_2,b)=(g,\mathit{ns}),~~~~\eta(T_2,T_3,b)=(\mathit{ns},u),
\end{gather*}
where by $u,g,z,$ and $\mathit{ns}$ we denote uniform, Gaussian, Zipfian, and non-specified distributions, respectively.
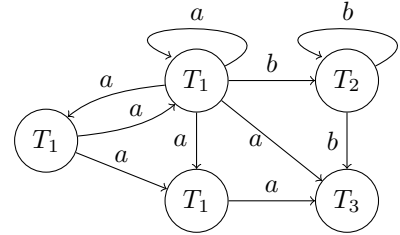
\begin{figure}[h]
\centering
\begin{tikzpicture}[yscale=0.4]
	\node[state] (v1) at (0, 0) {$T_1$};
	\node[state] (v2) at (2, 2) {$T_1$};
	\node[state] (v3) at (2, -2) {$T_1$};
	\node[state] (v4) at (4, 2) {$T_2$};
	\node[state] (v5) at (4, -2) {$T_3$};

	\path[->] (v1) edge [above,bend right=25] node [align=center]  {$a$} (v2)
	(v2) edge [above,bend right=25] node [align=center]  {$a$} (v1)
	(v1) edge [above] node [align=center]  {$a$} (v3)
	(v2) edge [left] node [align=center]  {$a$} (v3)
	(v3) edge [above] node [align=center]  {$a$} (v5)
	(v2) edge [above] node [align=center]  {$b$} (v4)
	(v2) edge [left] node [align=center]  {$a$} (v5)
	(v4) edge [left] node [align=center]  {$b$} (v5)
	(v2) edge [loop,above] node [align=center]  {$a$} (v2)
	(v4) edge [loop,above] node [align=center]  {$b$} (v4)
;
\end{tikzpicture}
\caption{\label{fig:graph:schema}Graph from Example~\ref{ex:schema:graph}.}
\end{figure}
For instance, the graph in Fig.~\ref{fig:graph:schema} can be generated by using this graph configuration.
Although a much larger graph is needed to observe the actual distributions, we refer to Section~\ref{sec:query-generation-selectivity} for further examples also handling distributions.\hfill\ensuremath{\square}
\end{example}

\subsection{Query workload generation}\label{sec:query-generation}

We next formally define {\em query workload configurations} (cf.\ Fig.~\ref{fig:gmark_arch}).  Towards this, we first
outline the query language supported by $\mathsf{gMark}$.  As motivated in Section
\ref{sec:relatedwork}, we focus on generating {\em unions of conjunctions of regular path queries} ($\mathsf{UCRPQ}$).  This fundamental query language covers many queries which
appear in practice, including the core constructs of SPARQL~1.1 queries,
Neo4j's Cypher queries, and many Datalog-based encodings~\cite{Baeza13}.

Recall that $\Sigma$ is a finite alphabet (cf.\ Definition \ref{def:schema}) and let $\Sigma^+ = \{a, a^-\mid a\in\Sigma\}$, where $a^-$ denotes the {\em inverse} of the edge label $a$.
Let $V = \{?x, ?y, \ldots\}$ be a set of variables and $n>0$. 
A {\em query rule} is an expression of the form
$$(\overline{?v}) \gets (?x_1, r_1, ?y_1), \ldots, (?x_n, r_n, ?y_n)$$ where:
for each $1\leq i\leq n$, it is the case that $?x_i, ?y_i \in V$;
$\overline{?v}$ is a vector of zero or more of these variables, the length of
which is called the {\em arity} of the rule; and, for each $1\leq i\leq n$, it
is the case that $r_i$ is a regular expression over $\Sigma^+$ using $\{\cdot,
+,*\}$ (i.e., {\em concatenation}, {\em disjunction}, and {\em Kleene star}).
Without loss of generality, we restrict regular expressions to only use
recursion (i.e., the Kleene star symbol $*$) at the outermost level.  
Hence, expressions can always be
written to take either the form $(P_1 + \cdots + P_k)$ or the form $(P_1 +
\cdots + P_k)^*$, for some $k>0$, where each $P_i$ is a {\em path expression}
i.e., a concatenation of zero or more symbols in $\Sigma^+$.  We refer to
the right-hand side of a query rule as the {\em body} of the query rule, each
subgoal $(?x_i, r_i, ?y_i)$ of the body as a {\em conjunct}, and to the
left-hand side as the {\em head} of the query rule.

A {\em query} $Q\in\mathsf{UCRPQ}$ is a finite non-empty set of query rules, each of the same arity. 
The semantics $Q(G)$ of evaluating $Q$ on a given graph $G$ (having
edge labels in $\Sigma$) is the standard one following that of unions of
conjunctive Datalog queries~\cite{Baeza13}, assuming standard {\em set-oriented} semantics.
In summary, a query is basically a collection of {\em query rules}, each rule
having several {\em conjuncts}, each conjunct having several {\em disjuncts}
whose paths have a certain {\em length}.

\begin{example}\label{example:query-encoding}\normalfont
Take the following $\mathsf{UCRPQ}$:
\begin{flalign*}
&(?x,?y,?z)\leftarrow (?x, (a\cdot b+ c)^*, ?y), (?y,a,?w), (?w, b^-, ?z)\\
&(?x,?y,?z)\leftarrow (?x, (a\cdot b+ c)^*, ?y), (?y,a,?z)
\end{flalign*}
This query selects nodes $x,y,z$ such that one can navigate between $x$ and $y$ with a path in the language of $(a\cdot b+ c)^*$, and moreover, can navigate between $y$ and $z$ with a path in the language of $a\cdot b^- +a$.
This query consists of two rules consisting of three conjuncts and two conjuncts, resp.
The conjuncts of the form $(?x, (a\cdot b+ c)^*, ?y)$ have two disjuncts (of length 2 and 1, resp.) and all other conjuncts have only one disjunct (of length 1).
\hfill\ensuremath{\square}\end{example}
We define {\em query size} as a tuple $$t = ([r_{\min}, r_{\max}],[c_{\min},
c_{\max}],[d_{\min},d_{\max}],[l_{\min},l_{\max}])$$ providing intervals of
minimal and maximal 
values for the number of rules, conjuncts, disjuncts, and length of the paths in
the query, resp.,
that generated queries should have. For example, the query from
Example~\ref{example:query-encoding} has size $([2,2],[2,3],[1,2],[1,2])$.
In $\mathsf{gMark}$, users can specify minimal and maximal values for all of these parameters; 
in turn, the query generation algorithm can assign
values that range in these intervals. 
For simplicity of presentation, we assume in the remainder that a query consists of only one rule.

\begin{definition}\label{definition:query:workload:configuration}
A {\em query workload configuration} is a tuple
$\mathbb{Q}=(\mathbb{G},\#q,\mathit{ar},f,e,p_{r},t)$ where $\mathbb{G}$ is a graph configuration, $\#q$ is the number of queries in the workload (defined on all instances of $\mathbb{G}$), 
$\mathit{ar}$ is the arity constraint, $f$ is the 
shape constraint, $e$ is the selectivity of
the queries in the workload, $p_{r}$ is the probability of recursion, and
$t$ is the query size.  
\end{definition}
Notice that in addition to the graph configuration $\mathbb{G}$ (cf.\ Section~\ref{sec:graph-generation}), the user can specify several other constraints.
First, $\mathit{ar}$ is the range of allowed arities for the queries in the workload. 
For instance, the query from
Example~\ref{example:query-encoding} has arity 3.  We also support Boolean
queries (arity 0). The shape constraint $f$ contains the supported query shapes
(among which chain, star, cycle, and star-chain are supported in
$\mathsf{gMark}$) and the user can specify which among them she would like to have in the
generated query workload. 
Similarly, the selectivity constraint $e$ contains the desired
selectivity classes, among which we support constant, linear and quadratic (cf.\ Section~\ref{sec:querygen}).  The
user can further specify the probability to have the multiplicity * above a
disjunct, reflected by the parameter $p_{r}$.

We finally point out that $\mathsf{gMark}$ is able to {\em translate} the generated $\mathsf{UCRPQ}$ in four concrete syntaxes (cf.\ Fig.~\ref{fig:gmark_arch}): SPARQL, openCypher, PostgreSQL, and Datalog.

\subsection{Intractability of the generation problem}\label{sec:intractability-generation}
In this section, we prove the {\em intractability} of the problems of graph and query workload generation.
First, we prove that the graph generation problem is intractable.

\begin{theorem}\label{thgenerationnpc}
Given a graph configuration $\mathbb{G}$, deciding whether there exists a graph satisfying $\mathbb{G}$ is NP-complete.
\end{theorem}

Prior to presenting the proof of Theorem~\ref{thgenerationnpc}, we would like to introduce some standard macros for encoding pairs of in- and out-distributions.
Precisely, we use:

$\bullet$ ``1'' for non-specified in-degree distribution and uniform out-degree distribution with min=max=1. 
In other words, $\eta(T_1,T_2,a)= 1$ means that from a node of type $T_1$ there is precisely one outgoing $a$-labeled edge to a node of type $T_2$, and that in a node of type $T_2$ we can have an arbitrary number of incoming $a$-labeled edges from nodes of type $T_1$.

$\bullet$ ``?'' for non-specified in-degree distribution and uniform out-degree distribution with min=0 and max=1.

$\bullet$ ``0'' for non-specified in-degree distribution and uniform out-degree distribution with min=max=0.

\begin{proof}\normalfont
We show the NP-hardness by reduction from the $\mathrm{SAT}_{1\mbox{-}\textrm{in}\mbox{-}3}$ {problem}, known to be NP-complete~\cite{Schaefer78}.
Take a 3CNF formula $\varphi=C_1\land\ldots\land C_k$ over variables $x_1,\ldots,x_n$.
We construct a schema $\mathcal{S}_\varphi = (\Sigma_\varphi, \Theta_\varphi, \mathcal{T}_\varphi, \eta_\varphi)$ and a subsequent graph configuration $\mathbb{G}_{\varphi} = (n_\varphi, \mathcal{S}_\varphi)$ as follows:

$\bullet$ $n_\varphi$: The graph should have $2\times n+ k + 1$ nodes.

$\bullet$ $\Sigma_{\varphi}$: There should be $3\times n+ k $ symbols (predicates) in the alphabet: $\Sigma_\varphi=\{c_1,\ldots,c_k,b_1,\ldots,b_n,t_1,f_1,\ldots,t_n,f_n\}$

$\bullet$ $\Theta_{\varphi}$: There should be $3\times n+ k + 1$ node types: $\Theta_\varphi=\{A,C_1,\ldots,C_k,B_1,\ldots,B_n,T_1,F_1,\ldots,T_n,F_n\}$.

$\bullet$ $\mathcal{T}_\varphi$: There should be precisely one node of type $A$ in the graph, which can be expressed as $\mathcal{T}_\varphi(A) = 1$.
Additionally, $\mathcal{T}_\varphi(B_1) = \ldots = \mathcal{T}_\varphi(B_n) = \mathcal{T}_\varphi(C_1) = \ldots = \mathcal{T}_\varphi(C_k) = 1$.
Notice that all these constraints could be alternatively expressed by saying that the proportion of each of them should be $1/(2\times n + k + 1)$ of the graph nodes.

$\bullet$ $\eta_{\varphi}$ such that:

\hspace*{0.35cm}$\bullet$ $\eta_{\varphi}(A, T_1, t_1) = \ldots = \eta_{\varphi}(A, T_n, t_n) = ?$

\hspace*{0.35cm}$\bullet$ $\eta_{\varphi}(A, F_1, f_1) = \ldots = \eta_{\varphi}(A, F_n, f_n) = ?$

\hspace*{0.35cm}$\bullet$ $\eta_{\varphi}(T_i, C_{l_1}, c_{l_1})\!=\!\ldots\!=\!\eta_{\varphi}(T_i, C_{l_m}, c_{l_m})\!=\!\eta_{\varphi}(T_i, B_{i}, b_i) = 1$, for every $i\in\{1,\ldots,n\}$, where $c_{l_1},\ldots, c_{l_m}$ correspond to the clauses in which the variable $x_i$ appears in a positive literal.

\hspace*{0.35cm}$\bullet$ $\eta_{\varphi}(F_i, C_{l_1}, c_{l_1})\! =\!\ldots\!=\!\eta_{\varphi}(F_i, C_{l_m}, c_{l_m})\!=\!\eta_{\varphi}(F_i, B_{i}, b_i) = 1$, for every $i\in\{1,\ldots,n\}$, where $c_{l_1},\ldots, c_{l_m}$ correspond to the clauses in which the variable $x_i$ appears in a negative literal.

\hspace*{0.35cm}$\bullet$ $\eta_{\varphi}(X,Y,b)=0$ for all other combinations of types $X, Y$ and symbols $b$ not present in one of the aforementioned cases.

We illustrate the construction for $\varphi_0 = (x_1\lor\neg x_2\lor x_3)\land (\neg x_1\lor x_3\lor \neg x_4)$, for which we obtain (we omit all cases with 0):

$\bullet$ $\eta_{\varphi}(A,T_1,t_1)=\eta_{\varphi}(A, F_1, f_1)=\ldots= \eta_{\varphi}(A,T_4,t_4)=$ 

$= \eta_{\varphi}(A, F_4, f_4) = ?$

$\bullet$ $\eta_{\varphi}(T_1,C_1,c_1) = \eta_{\varphi} (T_1,B_1,b_1) =\eta_{\varphi}(F_1,C_2,c_2)=$ 

$= \eta_{\varphi} (F_1,B_1,b_1) = \eta_{\varphi}(T_2, B_2, b_2) = \eta_{\varphi}(F_2, C_1,c_1)=$ 

$= \eta_{\varphi} (F_2, B_2, b_2) = \eta_{\varphi}(T_3, C_1, c_1) = \eta_{\varphi}(T_3, C_2, c_2)=$ 

$=\eta_{\varphi}(T_3, B_3, b_3) = \eta_{\varphi}(F_3, B_3, b_3) = \eta_{\varphi}(T_4,B_4,b_4) =$ 

$= \eta_{\varphi} (F_4, C_2,c_2) = \eta _{\varphi}(F_4, B_4,b_4) = 1$

\noindent We claim that $\varphi\in\mathrm{SAT}_{1\mbox{-}\textrm{in}\mbox{-}3}$ iff there exists a graph that satisfies the graph configuration.

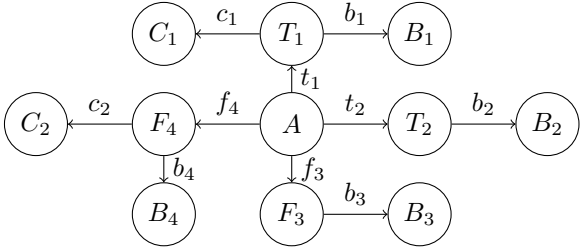
\begin{figure}[htb]\centering
\begin{tikzpicture}[xscale=1.7,yscale=1.2]
\node[state] (a) at (0,0) {$A$};
\node[state] (t1) at (0, 1) {$T_1$};
\node[state] (t2) at (1, 0) {$T_2$};
\node[state] (f3) at (0, -1) {$F_3$};
\node[state] (f4) at (-1, 0) {$F_4$};
\node[state] (c1) at (-1, 1) {$C_1$};
\node[state] (b1) at (1, 1) {$B_1$};
\node[state] (b2) at (2, 0) {$B_2$};
\node[state] (b3) at (1, -1) {$B_3$};
\node[state] (c2) at (-2, 0) {$C_2$};
\node[state] (b4) at (-1, -1) {$B_4$};
\path[->] 
(a) edge [right] node [align=center]  {$t_1$} (t1)
(a) edge [above] node [align=center]  {$t_2$} (t2)
(a) edge [right] node [align=center]  {$f_3$} (f3)
(a) edge [above] node [align=center]  {$f_4$} (f4)
(t1) edge [above] node [align=center]  {$c_1$} (c1)
(t1) edge [above] node [align=center]  {$b_1$} (b1)
(t2) edge [above] node [align=center]  {$b_2$} (b2)
(f3) edge [above] node [align=center]  {$b_3$} (b3)
(f4) edge [above] node [align=center]  {$c_2$} (c2)
(f4) edge [right] node [align=center]  {$b_4$} (b4)
;
\end{tikzpicture}
\caption{\label{fig:np-c:proof}Graph from the proof of Theorem~\ref{thgenerationnpc}.}
\end{figure}

For the {\em only if} part, take a valuation that satisfies exactly one literal of each clause and construct a graph that encodes this valuation, starting from a node of type $A$.
For example, for the above formula $\varphi_0$ and the valuation such that $x_1$ and $x_2$ are true, and $x_3$ and $x_4$ are false, construct the graph from Fig.~\ref{fig:np-c:proof}.
Since we choose exactly one among $T_i$ and $F_i$ for every $i\in\{1,\ldots, n\}$ we have exactly one $B_i$.
Moreover, since exactly one literal of each clause is satisfied we have exactly one $C_i$ (for $i\in\{1,\ldots, k\}$).
Thus, the constraints in $\mathcal{T}_\varphi$ are satisfied. 
As for the number of nodes, we have $2\times n$ (because there is exactly one valuation of each variable and we also have its corresponding $B_i$) + $k$ (because of the $k$ clauses) + 1 (the node of type $A$).
Consequently, the constructed graph satisfies all the constraints from the configuration.

For the {\em if} part, take a graph satisfying the constraints. 
Since it satisfies the constraints from $\mathcal{T}_\varphi$, the graph should have one node $A$, one $B_i$ (for $i\in\{1,\ldots, n\}$) and one $C_i$ (for $i\in\{1,\ldots,k\})$.
Since the total size of the graph is of $2\times n + k + 1$ nodes and seen how we can reach $B_i$'s and $C_i$'s based on the schema, we infer that the other $n$ nodes correspond to nodes of type $T_i$ or $F_i$ (encoding a valuation of the variable $x_i$).
Since we have precisely one $B_i$, we infer that each variable has exactly one valuation and since we have precisely one $C_i$, we infer that in each clause there is exactly one literal that is satisfied.
This means that the formula $\varphi$ is in $\mathrm{SAT}_{1\mbox{-}\textrm{in}\mbox{-}3}$.

To show the membership of the problem to NP, we point out that a non-deterministic Turing machine has to guess a graph having as many nodes as the constraint from the configuration.
The size of such a graph is thus polynomial in the size of the input and testing whether it satisfies the schema can be easily done in polynomial time.
\hfill\ensuremath{\square}\end{proof}

As a natural consequence of Theorem \ref{thgenerationnpc}, we have that the query workload generation problem is also intractable. 
That means that some parameters of the query workload cannot be fulfilled and it is not possible to test this efficiently. 
Hence, $\mathsf{gMark}$ follows a heuristic strategy in the generation.
More precisely, it attempts to achieve the exact values of the parameters and it may decide to relax some of them in order to obtain linear running time.
However, we do not claim that our algorithm is the best linear time algorithm.
Nonetheless, in the experimental study we observed fast generation time and good quality graph instances and queries.

\begin{corollary}\label{corollary-intractability}
Given a query workload configuration $\mathbb{Q}$, deciding whether there exists a query workload satisfying $\mathbb{Q}$ is NP-complete.
\end{corollary}

\begin{proof}\normalfont
Recall that a query workload configuration $\mathbb{Q}$ is a tuple $(\mathbb{G},\#q,\mathit{ar},f,e,p_{r},t)$.
Since the graph configuration $\mathbb{G}$ is part of the input of $\mathbb{Q}$, we can take for $\mathbb{G}$ precisely the same encoding of $\mathrm{SAT}_{1\mbox{-}\textrm{in}\mbox{-}3}$ as in the proof of Theorem~\ref{thgenerationnpc} and then arbitrary values for the other constraints in $\mathbb{Q}$.
Then, both {\em if} and {\em only if} parts for the NP-hardness follow precisely as in the proof of Theorem~\ref{thgenerationnpc}.
As for the membership of the problem to NP, we point out that a non-deterministic Turing machine has to guess a query workload and then we can easily decide in polynomial time whether it satisfies the input constraints.
In particular, the query workload is of polynomial size (since the size is given as parameter in the query workload configuration), and the graphs that we need to check whether the query workload satisfies the selectivity constraints are also of polynomial size (since the graph size is given as part of the graph configuration that is also part of the query workload configuration).
\hfill\ensuremath{\square}\end{proof}
Despite the intractability of the generation problems and the $\mathsf{gMark}$ heuristic approach, we would like to already point out that the $\mathsf{gMark}$ graph and query generator leads to highly accurate results, as we detail with our experiments in Section~\ref{subsect:scala}, where we report on the selectivity estimation of the generated queries over the generated graphs.

\section{Graph generation}\label{sec:graphgen}

As shown in Theorem~\ref{thgenerationnpc}, given a graph configuration $\mathbb{G}$, deciding whether there exists a graph satisfying $\mathbb{G}$ is {\em NP-complete}.
Consequently, it is not always possible to generate in PTIME a graph satisfying a given graph configuration $\mathbb{G}$.
Nevertheless, generating graph instances in PTIME is an important goal in designing a graph generator, essentially 
because of {\em scalability}. To achieve scalability we must therefore relax some constraints specified in the graph configuration.

For all such reasons, the {\em $\mathsf{gMark}$ generation algorithm} (cf.\ Fig.~\ref{alg:graph-generation}) takes a heuristic approach that guarantees {\em linear time} generation (i.e. linear in the size of the input and output).
The algorithm pseudocode, which is conceptually quite straightforward, is illustrated below.
For each constraint $\eta(T_1,T_2,a)$, it generates a set of edges as follows.
It creates a vector $v_{\mathit{src}}$ for the source nodes containing each of the nodes of type $T_1$ a number of times drawn according to the out-distribution $D_{\mathit{out}}$. It initializes (line 2), fills (lines 3-4), and, respectively, shuffles (line 7) this vector.
The length $n_{T_1}$ of this vector is computed either as $n\times\mathcal{T}(T_1)$ if $\mathcal{T}(T_1)$ is a proportion, or simply $\mathcal{T}(T_1)$ if $\mathcal{T}(T_1)$ is a fixed value.
The algorithm creates a vector for the target nodes $v_{\mathit{trg}}$ in a similar manner (lines 2, 5--7).
Next, it simultaneously iterates over the two vectors $v_{\mathit{src}}$ and $v_{\mathit{trg}}$ and outputs a number of $a$-labeled edges corresponding to the minimal length of the two vectors (lines 8-9).
The function $\mathit{id}_T(j)$ then returns the $j^{\mathrm{th}}$ node of type $T$ in the graph (assuming that the graph nodes are uniquely identified).
We note that the random draws in Lines 4 and 6 are statistically independent and hence
the order in which we process j has no impact. Likewise, the processing of constraints in Line 1
is statistically independent and hence order-independent.

\begin{figure}
\fbox{
\begin{minipage}{8.35cm}
{\bf Input}: A graph configuration $\mathbb{G} = (n,\mathcal{S} = (\Sigma, \Theta, \mathcal{T}, \eta))$\\
{\bf Output}: A set of edges (source node, label, target node)\\
\begin{tabular}{ll}
{\scriptsize 1:}&\hspace*{-0.25cm}{\bf foreach} $\eta(T_1,T_2,a)=(D_{\mathit{in}}, D_{\mathit{out}})$ \textbf{do}\\
{\scriptsize 2:}&\hspace*{-0.25cm}~~~\textbf{let} $v_{\mathit{src}}$ and $v_{\mathit{trg}}$ be two empty vectors\\
{\scriptsize 3:}&\hspace*{-0.25cm}~~~\textbf{for} $j\in \{1,\ldots, n_{T_1}\}$ \textbf{do}\\
{\scriptsize 4:}&\hspace*{-0.25cm}~~~~~~{\bf add} $\mathit{draw}(D_{\mathit{out}})$ occurrences of $j$ in $v_{\mathit{src}}$\\
{\scriptsize 5:}&\hspace*{-0.25cm}~~~\textbf{for} $j\in \{1,\ldots, n_{T_2}\}$ \textbf{do}\\
{\scriptsize 6:}&\hspace*{-0.25cm}~~~~~~{\bf add} $\mathit{draw}(D_{\mathit{in}})$ occurrences of $j$ in $v_{\mathit{trg}}$\\
{\scriptsize 7:}&\hspace*{-0.25cm}~~~{\em shuffle}$(v_{\mathit{src}})$; {\em shuffle}$(v_{\mathit{trg}})$\\
{\scriptsize 8:}&\hspace*{-0.25cm}~~~\textbf{for} $i\in\{1,\ldots,\min(|v_{\mathit{src}}|,|v_{\mathit{trg}}|)\}$ \textbf{do}\\
{\scriptsize 9:}&\hspace*{-0.25cm}~~~~~~{\bf output} $(\mathit{id}_{T_1}(v_{\mathit{src}}[i]), a, \mathit{id}_{T_2}(v_{\mathit{trg}}[i]))$
\end{tabular}
\end{minipage}}
\caption{\label{alg:graph-generation}Graph generation algorithm.}
\vspace*{-0.2cm}
\end{figure}

Algorithm from Fig.~\ref{alg:graph-generation} produces an output graph {\em in linear time}, without backtracking.
Our implementation leverages further optimizations, notably exploiting the average information of the Gaussian distributions to avoid entirely constructing the vectors when the in- and/or out-distributions are Gaussian.
Our experiments in Section~\ref{subsect:scala} confirm the practical efficiency and scalability of our graph generator.

Due to the heuristic nature of our algorithm, some of the constraints in the graph configuration might not be fulfilled in the output graph.
In particular, for each constraint it generates a number of edges equal to the minimal length of the two vectors (line~8).
Consequently, whenever the two vectors have different sizes, then the generated graph may contain nodes that do not satisfy the precise values dictated by the in- or out-distributions.
Notice that these distributions are essential in the proof of Theorem~\ref{thgenerationnpc} for showing the intractability of the generation problem.

Due to practical reasons, we decided that our generator should always return a graph to the user instead of performing a potentially costly satisfiability check and possibly aborting the generation.
Moreover, this choice still allows to globally preserve the types of distribution for each constraint $\eta(T_1,T_2,a)$, even though the generated number of edges may not satisfy the exact parameters of the in- or out-distributions.
In fact, our method relies on the types of distributions (uniform, Gaussian, Zipfian) and not on the actual parameters of the distributions. Such a choice of preserving the global distributions in the graph generator also turns to 
be useful in the query generator, where we adopt a novel technique of selectivity estimation for the generated queries.
As shown in the experiments in Section~\ref{subsect:scala}, the obtained selectivities exhibit high accuracy, thus confirming the 
effectiveness of our method.

\vspace*{-0.2cm}
\section{Query Generation}\label{sec:query-generation-selectivity}

A core innovation of $\mathsf{gMark}$ is that graph instances and query workloads
are both generated from graph schemas. This allows the tight coupling of
queries to instances while still also supporting fine-grained control of the
diversity of query workloads.  As discussed in Section~\ref{sec:query-generation}, query generation in $\mathsf{gMark}$ is guided by the same
input schema used for graph generation, which makes the queries of the workload
pertinent to graph instances. 

\begin{figure}\fbox{
\begin{minipage}{8.35cm}
{\bf Input}: A query workload configuration $\mathbb{Q}=$ \\
\hspace*{1.1cm}$(\mathbb{G}=(n,\mathcal{S}),\#q,\mathit{ar},f,e,p_r,t)$\\
{\bf Output}: A set of queries\\
\begin{tabular}{ll}
{\scriptsize 1:}&\hspace*{-0.25cm}\textbf{for} $i\in\{1,\ldots,\#q\}$ \textbf{do}\\
{\scriptsize 2:}&\hspace*{-0.25cm}~~~\textbf{let} $\mathit{skeleton=get\_query\_skeleton}(f,t)$\\
{\scriptsize 3:}&\hspace*{-0.25cm}~~~$\mathit{add\_projection\_variables}(\mathit{skeleton}, \mathit{ar})$\\
{\scriptsize 4:}&\hspace*{-0.25cm}~~~\textbf{let} $Q=\mathit{instantiate\_placeholders}(\mathit{skeleton}, \mathcal{S}, p_{r}, t)$\\
{\scriptsize 5:}&\hspace*{-0.25cm}~~~{\bf output} $Q$
\end{tabular}
\end{minipage}}
\caption{\label{alg:query-general}Query generation algorithm.}
\vspace*{-0.2cm}
\end{figure}

In Section~\ref{sec:general:algorithm}, we present our {\em query generation algorithm}, which ensures query diversity in terms of {\em arity}, {\em shape} (chain, star, cycle, star-chain), {\em recursion} (probability to have Kleene star above conjuncts), and {\em size} (number of conjuncts, disjuncts, length of paths).
We have already implemented in $\mathsf{gMark}$ this algorithm supporting the diversity of all the aforementioned parameters.

Additionally, users can specify constraints on the query {\em selectivity}.
In $\mathsf{gMark}$, we guarantee the selectivity estimation for binary queries, a natural and broad class in the context of graph queries, strictly containing the regular path queries (i.e., property paths in SPARQL~1.1). Selectivity estimation for such queries already involved the development of novel sophisticated machinery.
Indeed, in Section~\ref{sec:querygen} we formalize the notion of {\em selectivity classes for binary queries} and we explain the additional data structures and techniques needed to accommodate selectivity estimation for binary queries in the general algorithm.
The schema-driven techniques that we introduce for selectivity estimation are non-trivial and, to the best of our knowledge, novel.
Moreover, binary queries already make interesting practical cases for graph database benchmarking.
In particular, our study based on such queries points out important limitations of existing graph database engines, as we highlight in Section~\ref{sec:databasesStudy}.

\subsection{General algorithm}\label{sec:general:algorithm}

The {\em query generation algorithm} (cf.\ Fig.~\ref{alg:query-general}) takes as input a {\em query workload configuration} $\mathbb{Q}$ (cf.\ Section~\ref{sec:query-generation}) and outputs a set of queries.

For each query, we first generate a query skeleton for the body of the query (line~2) by considering the query shape and the number of conjuncts.
A query skeleton is a set of conjuncts of the form $(?x_1, P, ?x_2)$ where $?x_1$ and $?x_2$ are variables, and $P$ is a placeholder that we instantiate afterwards with a regular expression. 
For example, if we want a chain query and a number of conjuncts $c_{\min}\!=\!c_{\max}\!=\!3$, we obtain a query skeleton $(?x_1,P_1,?x_2), (?x_2,P_2,?x_3), (?x_3,P_3,?x_4)$. 

The other query shapes are essentially combinations of different {\em chain queries}.
More precisely, we implemented {\em cycle queries} as two chain queries that share the same endpoint variables, {\em star queries} as combinations of chains that have
the same starting variable, and {\em star-chain queries} as combinations of chain queries and star queries. 

Next, we randomly pick a set of projection variables such that their number is consistent with the arity constraint (line~3); in $\mathsf{gMark}$, we already support the generation of queries of arbitrary arity (including Boolean i.e., arity~0).
We finally instantiate the placeholders (line~4) with regular  expressions that satisfy the probability of recursion and the other aspects of the size (number of disjuncts and length of paths), and we output the obtained query (line~5).

Algorithm from Fig.~\ref{alg:query-general} works in linear time in the size of the input and output.
It produces a query workload coupled to the graph schema, and satisfying the arity, shape, recursion and size constraints.
As presented here, it disregards the selectivity constraint $e$ that is the trickier parameter of a graph configuration and treated in the next section.

\subsection{Selectivity estimation algorithm}\label{sec:querygen}

In this section, we present an innovative method to estimate the selectivity of binary queries, that we already implemented in $\mathsf{gMark}$.
Our method impacts the step from line~4 in Fig.~\ref{alg:query-general}, and requires new developments.
In Section~\ref{sec:trihotomy}, we classify the binary queries in three classes depending on the asymptotic behavior of the size of their result.
In Section~\ref{sec:estimating}, we propose a method to estimate the selectivity.
In Section~\ref{sec:data:structures}, we introduce the data structures needed to implement it.
In Section~\ref{sec:alg:revisited}, we revisit the algorithm from Fig.~\ref{alg:query-general} to accommodate selectivity estimation.

\subsubsection{Constant, linear, and quadratic queries}\label{sec:trihotomy}
The {\em selectivity} of a query $Q$ on a graph $G$ is the number
of results returned by the evaluation of $Q$ on $G$,  i.e., $|Q(G)|$. 
This number depends on both the topology and the actual size of the graph instance, given that instances of different sizes can be generated upon the same input schema. 

Given a binary query $Q$ and a schema $\mathcal{S}$, for every graph $G$ satisfying $\mathcal{S}$, we assume that the value $|Q(G)|$ behaves asymptotically as a function of the form $|Q(G)| = \beta |G|^\alpha$, where $\alpha$ and $\beta$ are real constants\footnote{Our experimental study duly confirms this assumption.}.
We say that the above value $\alpha$ is the {\em selectivity value of $Q$} ({\em w.r.t.\ }$\mathcal{S}$), denoted $\alpha_\mathcal{S}(Q)$ or simply $\alpha(Q)$ when it does not lead to ambiguity.
Thus, the selectivity value of a query is by definition bounded by the query arity.
Since we focus on binary queries, we consider selectivity values such that $0 \leq \alpha(Q)\leq 2$. 

We identify three practical query classes, depending on whether $\alpha(Q)$ is closer to $0$, $1$, or $2$:

$\bullet$ {\em Constant queries} (for which $\alpha(Q)\approx 0$) select a number of results that does not grow (or barely grows) with the graph size.
For instance, a query selecting pairs ({\tt country}, {\tt language}) is constant if the graphs follow a realistic schema specifying that the numbers of countries and languages do not grow with the graph size, and hence the number of query results is more or less constant.

$\bullet$ {\em Linear queries} (for which $\alpha(Q)\approx 1$) select a number of results that grows at a rate close to the growth of the number of nodes in the graph instances.
For example, a query selecting pairs ({\tt language}, {\tt user}) is linear if the schema specifies that the number of users grows with the graph, whereas the number of languages is more or less constant.
Another example of a linear query is ({\tt user}, {\tt address}) if we assume that the schema specifies that each user has precisely one address and the number of users grows linearly with the graph.

$\bullet$ {\em Quadratic queries} (for which $\alpha(Q)\approx 2$) select a number of results that grows at a rate close to the growth of the square of the number of nodes in the graph instances.
For example, the transitive closure of the {\tt knows} predicate in a social network is quadratic because a realistic schema should specify that this predicate follows a power-law (e.g., Zipfian) in- and out-distribution.
Thus, the query result contains Cartesian products of subsets of users that know and are known by some hub users of the social network.

\subsubsection{Estimating the selectivity value}
\label{sec:estimating}
We propose a solution for estimating the selectivity value $\alpha(Q)$ of a given query $Q$, for all graphs satisfying a given schema $\mathcal{S}$.
This basically means to compute a function that associates to $Q$ a value $\hat{\alpha}(Q) \in \{ 0, 1, 2\}$. 
This value can be made more precise as follows: for a pair of node
types $A$ and $B$,  $\hat{\alpha}_{A,B}(Q)$ is the estimated selectivity of $Q$
restricted to pairs $(x,y)$ where $x$ is of type $A$ and $y$ of type $B$. 
Then, the overall estimated selectivity value of $Q$ is 
$\hat{\alpha}(Q) =
\displaystyle \max_{A,B} (\hat{\alpha}_{A,B}(Q))$. 

To compute these values, we define an algebra based on what we call {\em
selectivity classes}. First, for each node type $A$ within the input schema
$\mathcal{S}$, we denote $\mathsf{Type}(A) = N$ if $A$ grows with the graph size and $\mathsf{Type}(A) = 1$ if it does not. In the graph schema, $\mathsf{Type}(A) = 1$ if $\mathcal{T}(A)$ is a fixed value and $\mathsf{Type}(A) = N$ if $\mathcal{T}(A)$ is a proportional value.

For each query $Q$ and each pair of node types $A$ and $B$, the {\em selectivity
class of $Q$ for} $A$, $B$, denoted $\mathit{sel}_{A,B}(Q)$ is a {\em triple} $(t_A, o, t_B)$ such
that $t_A = \mathsf{Type}(A)$, $t_B = \mathsf{Type}(B)$ and $o \in \{ =, <, >, \diamond, \times
\}$ is an operation between types.

\begin{table}\scriptsize
\caption{\label{tbl:symbols}Algebraic operations between types.}
\vspace*{-0.3cm}
\begin{tabular}{|c|c|c|c|}\hline
{\em \hspace*{-0.15cm}Operation\hspace*{-0.1cm}} & $|\{ n \mid (n_1, n) \in Q(G)\}|$ & $|\{n \mid (n, n_2) \in Q(G)\}|$ & $\alpha(Q)$\\\hline
$=$ & Bounded & Bounded & 0 or 1 \\\hline
$<$ & Bounded & Not bounded & 1\\\hline
$>$ & Not bounded & Bounded & 1 \\\hline
$\diamond$ & Not bounded & Not bounded & 1\\\hline
$\times$ & Not bounded & Not bounded & 2\\\hline
\end{tabular}
\end{table}

We summarize these algebraic operations in Table~\ref{tbl:symbols}, which should be read a follows: 
an operation from the first column denotes that for every graph $G$ satisfying a schema $\mathcal{S}$, for every pair of nodes $(n_1, n_2) \in Q(G)$, it is the case that $|\{ n \mid (n_1, n)  \in Q(G)\}|$ and $|\{n \mid (n, n_2) \in Q(G)\}|$, resp., are or are not bounded (by some constants), as indicated in the second and third columns, resp.
The last column $\alpha(Q)$ is particularly useful to distinguish between the last two operations $\diamond$ and $\times$.

We next intuitively explain the above operations and we illustrate them via examples:

$\bullet$ $=$ is the simplest operation and occurs either (i) between
constant types e.g., ({\tt country}, {\tt language}) as illustrated for
constant queries, or (ii) for some linear queries such as the query defined
by the empty regular expression $\varepsilon$ that returns precisely as
many results as the nodes in the graph.

$\bullet$ $<$ characterizes queries where either (i) the out-degree distribution is Zipfian, such as ({\tt language}, {\tt user}) as illustrated for linear queries, or (ii) the source node type $A$ has $\mathsf{Type}(A)=1$ and the target node type $B$ has $\mathsf{Type}(B)=N$.
Then, the definition of $>$ is symmetric to $<$. 

$\bullet$ $\times$ corresponds to queries performing a Cartesian product between two node sets (both growing with the graph), for example the transitive closure of the {\tt knows} predicate that we used above to illustrate the quadratic queries.
Intuitively, the $\times$ is the result of a $>$ followed by a $<$.

$\bullet$ $\diamond$ is the trickier operation and corresponds for instance to pairs of users that are known by someone in common.
Most users are not linked in this way, but two pairs of hub users are.
Thus, although there are numerous paths between two hubs, their number remains relatively small, in particular it grows linearly with the graph.
Intuitively, the $\diamond$ is the result of a $<$ followed by a $>$.

As defined in Section~\ref{sec:query-generation}, the $\mathsf{gMark}$ workloads consist of queries expressed as $\mathsf{UCRPQ}$'s. Hence, we need to compute the selectivity values for regular path queries, which involve regular expressions. 
First, for a 
query $Q$ defined by the regular expression $\varepsilon$ (the empty word), for each type $A$, we have that $\mathit{sel}_{A,A}(Q) = (\mathsf{Type}(A), =, \mathsf{Type}(A))$.
When $Q$ is defined by a single edge label $a\in\Sigma^+$, we obtain $\mathit{sel}_{A,B}(Q)$ directly from the distribution of the $a$-labeled edges from $A$ to $B$, as defined above and obtained from the schema.

\begin{example}
\normalfont\label{example:types}
    Consider the schema given in Example~\ref{ex:schema:graph}. First, we assign selectivity classes to the types of the schema, thus $\mathsf{Type}(T_1) = \mathsf{Type}(T_2) = N$ whereas $\mathsf{Type}(T_3) = 1$. From the schema, we compute the following values:

$\bullet$ $\mathit{sel}_{T_1,T_1}(a) = (N,<,N)$ and $\mathit{sel}_{T_1,T_1}(a^-) = (N,>$, $N)$ (because of the Zipfian out-distribution that moreover implies a Zipfian in-distribution for the inverse, and $\mathsf{Type}(T_1)=N$), 

$\bullet$ $\mathit{sel}_{T_1,T_2}(b) = (N,=,N)$ and $\mathit{sel}_{T_2,T_1}(b^-) = (N,=,N)$ (because of non-Zipfian in- and out-distributions, and moreover, both $\mathsf{Type}(T_1)=\mathsf{Type}(T_2)= N$), 

$\bullet$ $\mathit{sel}_{T_2,T_2}(b) = (N,=,N)$ and $\mathit{sel}_{T_2,T_2}(b^-) = (N,=,N)$ (same reasoning as for the previous bullet),

$\bullet$ $\mathit{sel}_{T_2,T_3}(b) = (N,>,1)$ and $\mathit{sel}_{T_3,T_2}(b^-) = (1,<,N)$ (because of non-Zipfian in- and out-distributions, and moreover, $\mathsf{Type}(T_2)=N$ and $\mathsf{Type}(T_3)=1$).
\hfill\ensuremath{\square}\end{example}

Let a query $Q$ be defined by the regular expression $p_1 + p_2$ where $p_1$
and $p_2$ are two regular expressions that define queries $Q_1$ and
$Q_2$, respectively. For every pair of node types $A$, $B$, such that $\mathit{sel}_{A,B}(Q_1) =
(t_A, o_1, t_B)$ and $\mathit{sel}_{A,B}(Q_2) = (t_A, o_2, t_B)$ then $\mathit{sel}_{A,B}(Q)
= (t_A, o_1 + o_2, t_B)$ where $o_1 + o_2$ is defined by the table in
Fig.~\ref{algebra:disjunction}. 
Quite similarly, for $Q$ defined by $p_1 \cdot p_2$, we have $\mathit{sel}_{A,B} (Q) = \Sigma_{C \in \Theta} \mathit{sel}_{A,C}(Q_1) \cdot \mathit{sel}_{C,B}(Q_2)$
where $\mathit{sel}_{A,C}(Q_1) \cdot \mathit{sel}_{C,B}(Q_2) = (t_A, o_1 \cdot o_2, t_B)$
for $\mathit{sel}_{A,C}(Q_1) = (t_A, o_1, t_C)$ and  $\mathit{sel}_{C,B}(Q_2) = (t_C, o_2, t_B)$,
and where $o_1 \cdot o_2$ is defined by the table in Fig.~\ref{algebra:conjunction}. Finally, if $Q$ is defined by the regular expression $p^*$ where $p$ defines a query $Q'$, we assign a selectivity class to $Q$ if and only if the input and output types of $Q'$ are the same, in which case $\mathit{sel}_{A,A}(Q) = \mathit{sel}_{A,A}(Q') \cdot \mathit{sel}_{A,A}(Q')$.

\begin{figure}[t]
\subfigure[\label{algebra:disjunction}Disjunction.]{
\centering
$\begin{array}{|c|ccccc|}
\hline
+  & =  & <  & >  & \diamond & \times \\
\hline
=  & =  & <  & >  & \diamond & \times \\
<  & <  & <  & \diamond & \diamond & \times \\
>  & >  & \diamond & >  & \diamond & \times \\
\diamond & \diamond & \diamond & \diamond & \diamond & \times \\
\times  & \times  & \times & \times & \times  & \times \\
\hline 
\end{array}$
}
\hspace{0.5cm}
\subfigure[\label{algebra:conjunction}Conjunction.]{
\centering
$\begin{array}{|c|ccccc|}
\hline
\cdot   & =  & <  & >  & \diamond & \times \\
\hline
=  & =  & <  & >  & \diamond & \times \\
<  & <  & <  & \times  & \times  & \times \\
>  & >  & \diamond & >  & \diamond & \times \\
\diamond & \diamond & \diamond & \times  & \times  & \times \\
\times  & \times  & \times  & \times  & \times  & \times \\
\hline
\end{array}$
}
\vspace*{-0.2cm}
\caption{\label{fig:twotab}Algebra for selectivity classes.
To read in a (column, row) order.}
\vspace*{-0.2cm}
\end{figure}

As a remark, if either $t_A$ or $t_B$ is $1$, the operator solely relies on the other one.
Hence, the triples $(1, \times, 1)$ and $(1, \diamond, 1)$ are not permitted, which makes $(1,=,1)$, $(1,<,N)$ and $(N,>,1)$ the only permitted triples that contain a $1$.
However, in the computation of the algebraic expression, we could still
obtain triples $(1, \times, 1)$ and $(1, \diamond, 1)$, that we should
replace with $(1,=,1)$ if the case occurs.

Finally, the {\em estimated selectivity value} of a query is obtained directly
from its class. If the obtained selectivity class of a query $Q$
is $(1,=,1)$, then $\hat{\alpha}(Q) = 0$, if we obtain $(N,\times,N)$, then
$\hat{\alpha}(Q) = 2$, and $\hat{\alpha}(Q) = 1$ for all other cases.

\subsubsection{Data structures}\label{sec:data:structures}

The data structures needed to estimate the selectivity are: (a) the schema graph $\mathbb{G}_{\mathcal{S}}$, (b) the distance matrix $\mathcal{D}$, and (c) the selectivity graph $\mathbb{G}_{\mathit{sel}}$.

{\smallskip\noindent\bf (a) Schema graph $\mathbb{G}_{\mathcal{S}}$.}
We derive the schema graph from the schema $\mathcal{S}$.
Each node in $\mathbb{G}_{\mathcal{S}}$ is a pair given by a node type of the schema and a selectivity triple (cf.\ Section~\ref{sec:estimating}) associated to that type. 
More formally, the set of nodes of the schema graph $\mathbb{G}_{\mathcal{S}}$, denoted $\mathit{SelType}(\mathcal{S})$, consists of tuples $(T,(t_1, o, \mathsf{Type}(T)))$, where (i) $T$ is a node type from $\Theta$ and (ii) $(t_1, o, \mathsf{Type}(T))$ is a selectivity triple in the set of all possible selectivity triples $\{(1,=,1), (1,<,N), \ldots\}$.
The edges of $\mathbb{G}_{\mathcal{S}}$ are labeled with symbols in $\Sigma^+$.

The goal of the schema graph $\mathbb{G}_{\mathcal{S}}$ is to indicate how a path ending with a type $T$ of selectivity triple $(t_1, o, \mathsf{Type}(T))$ changes when it is extended with an edge in $\Sigma^+$. 
Formally, given a node $(T,(t_1, o, \mathsf{Type}(T)))$ and a label (or label inverse) $a\in\Sigma^+$ such that the schema allows an $a$-labeled edge between $T$ and a node type $T'$, if according to our algebra we have $(t_1, o, \mathsf{Type}(T)) \cdot \mathit{sel}_{T,T'}(a) = (t_1,o',\mathsf{Type}(T'))$, then in the schema graph $\mathbb{G}_{\mathcal{S}}$ there is an edge $((T,(t_1, o, \mathsf{Type}(T))), a, (T',(t_1,o',$ $\mathsf{Type}(T')))$.

\begin{example}
\normalfont
Recall Examples~\ref{ex:schema:graph} and~\ref{example:types}.
We illustrate a snippet of the corresponding schema graph in Fig.~\ref{fig:schemagraph}.
For instance, the node $(T_1,(N,=,N))$ is due to the fact that $\mathsf{Type}(T_1) = N$ and recall from Section~\ref{sec:estimating} that for a given type $A$ we have $\mathit{sel}_{A,A}(\varepsilon) = (\mathsf{Type}(A), =, \mathsf{Type}(A))$.
Moreover, our schema allows an $a$-labeled edge between two nodes of type $T_1$, following a Zipfian out-degree distribution, hence its selectivity triple is $(N, <, N)$, which explains the node $(T_1,(N,<,N))$ in Fig.~\ref{fig:schemagraph}.
Additionally, there is an $a$-labeled edge in our schema graph between the nodes $(T_1, (N,=,N))$ and $(T_1, (N,<,N))$ because in our algebra $(N,=,N)\cdot (N,<,N) = (N,<,N)$.  
\hfill\ensuremath{\square}\end{example}

\begin{figure}[t]
\centering
\begin{tikzpicture}[yscale=0.6]
\tikzstyle{fake}=[rectangle,draw,dotted]
\tikzstyle{node}=[rectangle,draw,rounded corners=3pt]

	\node[node] (v1) at (-3, 2) {\tiny$T_1,(N,<,N)$};
	\node[node] (v2) at (0, 2) {\tiny$T_1,(N,\diamond,N)$};
	\node[fake] (v3) at (2, 2) {};
	\node[node] (v4) at (-3, 0) {\tiny$T_1,(N,=,N)$};
	\node[node] (v5) at (0, 0) {\tiny$T_2,(N,=,N)$};
	\node[node] (v6) at (3, 0) {\tiny$T_3,(N,>,1)$};
	\node[fake] (v7) at (-3, -2) {};
	\node[fake] (v7b) at (-2, -2) {};
	\node[node] (v8) at (0, -2) {\tiny$T_2,(N,\times,N)$};

	\path[->] 
   (v4) edge [left] node [align=center]  {$a$} (v1)
	(v1) edge [above] node [align=center]  {$a^-$} (v2)
	(v2) edge [above] node [align=center]  {$a$} (v3)
	(v4) edge [above,bend left=15] node [align=center]  {$b$} (v5)
	(v5) edge [below,bend left=15] node [align=center]  {$b^-$} (v4)
	(v5) edge [above] node [align=center]  {$b$} (v6)
	(v6) edge [right, bend left=15] node [align=center]  {$b^-$} (v8)
	(v8) edge [left, bend left=15] node [align=center]  {$b$} (v6)
	(v8) edge [below] node [align=center]  {$b^-$} (v7b)
	(v4) edge [left] node [align=center]  {$a^-$} (v7)

	(v2) edge [loop,above] node [align=center]  {$a^-$} (v2)
	(v1) edge [loop,above] node [align=center]  {$a$} (v1)
	(v5) edge [loop,above] node [align=center]  {$b$} (v5)
;
\end{tikzpicture}
\vspace*{-0.2cm}
\caption{\label{fig:schemagraph}A snippet of the schema graph for our running example.}
\vspace*{-0.2cm}
\end{figure}
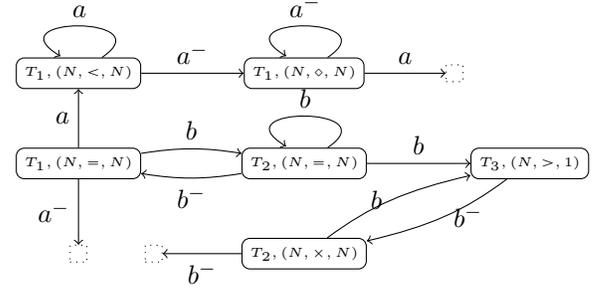

{\smallskip\noindent\bf (b) Distance matrix $\mathcal{D}$.}
The distance matrix $\mathcal{D}$  establishes for each pair 
$n, n' \in \mathit{SelType}(\mathcal{S})$ the length  $\mathcal{D}(n,n')$ of the shortest path between $n$ and $n'$ in $\mathcal{D}(n, n')$ in $\mathbb{G}_{\mathcal{S}}$.

{\smallskip\noindent\bf (c) Selectivity graph $\mathbb{G}_{\mathit{sel}}$.}
The selectivity graph $\mathbb{G}_{\mathit{sel}}$ is an unlabeled directed graph whose nodes are $\mathit{SelType}(\mathcal{S})$. An edge exists between two nodes $n$ and $n'$ if there exists a path between $n$ and $n'$ of length in $[l_{\min}, l_{\max}]$ (recall that this interval is specified as part of the query size in the workload configuration $\mathbb{Q}$). 

\begin{example}
\normalfont
We illustrate the selectivity graph for our running example in Fig.~\ref{fig:selectivitygraph}. 
As an example, there exists a path between $T_1,(N,=,N)$ and $T_2,(N,\times,N)$ (for instance, the path $b \cdot b \cdot b^-$), whose length is less than $l_{\max} = 4$. However, there does not exist such a path between $T_2,(N,\times,N)$ and $T_1,(N,=,N)$, thus there is no edge between those nodes.
\hfill\ensuremath{\square}\end{example}

\begin{figure}[t]
\centering
\begin{tikzpicture}[yscale=0.5]
\tikzstyle{fake}=[rectangle,draw,dotted]
\tikzstyle{node}=[rectangle,draw,rounded corners=3pt]

	\node[node] (v1) at (-3, 2) {\tiny$T_1,(N,=,N)$};
	\node[node] (v2) at (0, 2) {\tiny$T_2,(N,=,N)$};
	\node[node] (v3) at (-3, 0) {\tiny$T_3,(N,>,1)$};
	\node[node] (v4) at (0, 0) {\tiny$T_2,(N,\times,N)$};

	\path[->] 
   (v1) edge [bend left=15] node [align=center]  {} (v2)
   (v2) edge [bend left=15] node [align=center]  {} (v1)
   (v1) edge [] node [align=center]  {} (v3)
   (v1) edge [] node [align=center]  {} (v4)
   (v2) edge [] node [align=center]  {} (v3)
   (v2) edge [] node [align=center]  {} (v4)
   (v3) edge [bend left=15] node [align=center]  {} (v4)
   (v4) edge [bend left=15] node [align=center]  {} (v3)

	(v1) edge [loop,above] node [align=center]  {} (v1)
	(v2) edge [loop,above] node [align=center]  {} (v2)
	(v3) edge [loop,above] node [align=center]  {} (v3)
	(v4) edge [loop,above] node [align=center]  {} (v4)
;
\end{tikzpicture}
\vspace*{-0.2cm}
\caption{\label{fig:selectivitygraph}An excerpt of the selectivity graph
derived from the schema graph in Fig.~\ref{fig:schemagraph}.}
\vspace*{-0.2cm}
\end{figure}
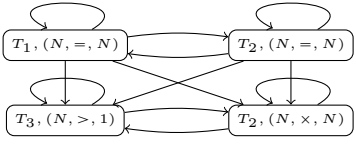

\subsubsection{Accommodating selectivity in query generation}\label{sec:alg:revisited}

We show how the instantiation of the placeholders (line~4 in Fig.~\ref{alg:query-general}) should be implemented to support selectivity estimation for binary queries.
Assume that we have already generated the query skeleton and picked the two projection variables (lines~2-3).

We first build a function $T$ that associates to each placeholder $P$ of the query skeleton a {\em selectivity type} $(T_1, (\mathsf{Type}(T_1),o,\mathsf{Type}(T_2)), T_2)$, where $T_1$ and $T_2$ are schema types, and $o$ is a selectivity operator. This function exhibits two properties. First, the input and output type of each selectivity operator should be consistent with the types of $T_1$, $T_2$ and secondly, selectivity values for each type should be selected in a way to guarantee that the global selectivity class for the query is the expected one.
To achieve this, we randomly choose a path on the selectivity graph $\mathbb{G}_{\mathit{sel}}$. 
This basically means that we want to find a path between a node with selectivity triple  $(?,=,?)$ (where by ``?'' we denote any type) to a node with  $(T_1, o, T_2)$ yielding one of the desired selectivities. The length of this path should be consistent with the configuration.
Some of the placeholders contain Kleene stars, with a probability dictated by the configuration.
Such conjuncts inherit the input and output types of their neighbor conjuncts, with the selectivity operator '='.

\begin{example}
\normalfont
In our running example, assuming that we look for a linear query with $3$ conjuncts, we can instantiate the function $T$ as follows: 
\begin{gather*}
T(P_1) = (T_1, (N,=,N), T_1), ~~~
T(P_2) = (T_1, (N,>,N), T_2),\\
T(P_3) = (T_2, (N,=,N), T_2).
\end{gather*}
We can then compute the selectivity of the concatenation, which is $(N,=,N) \cdot (N, >, N) \cdot (N, =, N) = (N, >, N)$, which corresponds to a linear query.
\hfill\ensuremath{\square}\end{example}

Note that drawing uniformly at random paths of a certain length in $\mathbb{G}_{\mathit{sel}}$ can be done efficiently with a two-step algorithm: first, 
each node $n$ is associated with a function $\mathit{nb\_path}(n,i)$ that gives the
number of paths of length $i$ that can be generated starting from $n$. For
instance, for a quadratic query, a node $n_1 : (A, (N, \times, N))$ has $\mathit{nb\_path}(n_1,0) = 1$
whereas a node $n_2 : (B, (1 = 1 ))$ has $\mathit{nb\_path}(n_2,0) = 0$. Other
values are obtained by a saturation algorithm: to generate a path of
length $l$, the algorithm picks a starting node with a random draw weighted
by $\mathit{nb\_path}(n,l)$, and then picks the label of an outgoing edge  to
a node $n'$ with a random draw weighted by $\mathit{nb\_path}(n', l-1)$, etc.
until all the nodes are saturated.

Next, for each placeholder we build a path skeleton that satisfies the constraints concerning the number of disjuncts and the length of the paths in the query size, as illustrated by the following example.

\begin{example}
\normalfont
To continue our running example, with a number of disjuncts in the range $[d_{\min},d_{\max}]=[3,5]$ and path length in the range $[l_{\min},l_{\max}]=[2,4]$, we may have a path skeleton as follows: $P_1 = X_1 \cdot X_2 + X_3 \cdot X_4 \cdot X_5 + X_6 \cdot X_7$ (i.e., three disjuncts having path length from 2 to 3). 
\hfill\ensuremath{\square}\end{example}
Once a path skeleton is computed, we need to find actual edge labels for each path, by randomly choosing paths in $\mathbb{G}_{\mathit{sel}}$. For the sake of conciseness, we omit the details and we illustrate it 
on the following example. 

\begin{example}
\normalfont
Let us consider $P_1 = X_1 \cdot X_2 + X_3 \cdot X_4 \cdot X_5 + X_6 \cdot X_7$, with $T(P_1) = (T_1, (N,=,N), T_1)$. The query workload generation algorithm may obtain a path instantiation of this kind: $X_1= b$ and $X_2 = b^-$ as the path $b \cdot b^-$ can go from $T_1$ to $T_2$ and from $T_2$ to $T_1$ via a concatenation of two paths with edge label $b$. 
\hfill\ensuremath{\square}\end{example}
The total size of the data structures described in Section~\ref{sec:data:structures} is quadratic in the size of the graph schema (due to the distance matrix) and the running time of the selectivity estimation algorithm inherits the quadratic behavior.
However, we observed that the query workload generation is very efficient, as we point out in more details at the end of Section~\ref{subsect:scala} for our default schemas and for the three existing schemas that we encoded in $\mathsf{gMark}$.
Moreover, it may not be possible that all constraints are satisfied at the same time and testing satisfiability is intractable (cf.\ Section~\ref{sec:intractability-generation}).
Our algorithm always outputs a result to the user.
More precisely, when instantiating the placeholders, it may be the case that we cannot fill the path skeleton with paths of the precise lengths and meeting the required selectivities, hence we choose to relax the path length in order to ensure accurate selectivity estimation and efficiency (instead of backtracking and drawing new skeletons).

\vspace*{-0.1cm}

\section{Empirical evaluation of gmark} \label{sec:gmarkStudy}
In this section, we empirically evaluate $\mathsf{gMark}$ w.r.t.\ two important aspects: the capability to encode the application domain of existing benchmarks (Section~\ref{subsec:scenarios}), and its quality in terms of both accuracy of the estimated selectivities and scalability of the generator (Section~\ref{subsect:scala}).

{\smallskip\noindent\bf Environment.} 
All experiments reported in this section and in Section~\ref{sec:databasesStudy} were run on an Intel Core i7 920, with 6GB RAM, and running Ubuntu 14.04 64bit.

\subsection{Coverage of practical graph scenarios}\label{subsec:scenarios}
In our experiments, we relied on four use cases: the default $\mathsf{gMark}$ use case, and three $\mathsf{gMark}$ encodings of the schemas of existing state-of-the-art benchmarks~\cite{AlucHOD14, Erling2015, SchmidtHLP09} that were possible since $\mathsf{gMark}$ can be easily tuned to fit an arbitrary set of predicates, node types, and schema constraints.

$\bullet$ $\texttt{Bib}$ is our default scenario, describing the bibliographical database introduced in Section~\ref{sec:motivating-example} as our motivating example.
It represents a baseline, illustrating the main $\mathsf{gMark}$ features, in particular all types of degree distributions. 

$\bullet$ $\texttt{LSN}$ is our $\mathsf{gMark}$ encoding of the fixed schema provided with the LDBC Social Network Benchmark~\cite{Erling2015}, which simulates user activity in a social network.

$\bullet$ $\texttt{SP}$ is our $\mathsf{gMark}$ encoding of the fixed DBLP-based~\cite{dblp} schema provided with SP2Bench~\cite{SchmidtHLP09}.

$\bullet$ $\texttt{WD}$ is our $\mathsf{gMark}$ encoding of the default schema provided with Waterloo SPARQL Diversity Test Suite (WatDiv)\,\cite{AlucHOD14}, which differs from LDBC and SP2Bench, and is similar to $\mathsf{gMark}$ in the sense that it also supports user-defined schemas via a data set description language.
The default schema that we encoded is about users and products.

The difference between $\mathsf{gMark}$ and the aforementioned benchmarks resides in the kind of expressible schema constraints that are incomparable.
Nonetheless, we have been able to encode their key characteristics, which include for instance node types, edge labels, and associations between entities. By opposite, we could not encode other characteristics that are typical of those benchmarks, such as subtyping and hardcoded correlations that we do not support in $\mathsf{gMark}$. 
We detail in Appendix~A in our technical report~\cite{BBCFLA16} the expressiveness differences between $\mathsf{gMark}$ and existing benchmarks, and our encoding choices of their sche\-mas.

{\smallskip\noindent\bf Discussion on the query loads.} 
Because of these differences, the $\mathsf{gMark}$ query loads are different from the fixed query loads in the other benchmarks. 
Nevertheless, although $\mathsf{gMark}$ does not generate exactly the same queries of the other benchmarks, it can easily be tuned to generate queries with similar characteristics, i.e., queries of comparable query shape and size and featuring the same selectivity of the original queries of those benchmarks. 
To concretely point out that this is the case, we report in Fig.~\ref{fig:plots0} an experiment in which we considered the execution times of three queries (one per  selectivity class) from the original SP2Bench query load and 
three comparable queries of the same shape, size and selectivity generated with $\mathsf{gMark}$ using $\texttt{SP}$. 
We observe that the queries generated and executed with $\mathsf{gMark}$ show the same asymptotic runtime behavior of the original SP2\-Bench queries executed under SP2Bench. Precisely, we can notice in Fig.~\ref{fig:plots0} that the obtained simulated constant (linear and quadratic, respectively) query in $\mathsf{gMark}$ falls in the same selectivity class of the original constant (linear and quadratic, respectively) query in SP2Bench.
Since the goal of $\mathsf{gMark}$ is not to simulate the exact query loads of other existing benchmarks, in the rest of the experimental study we only rely on the $\mathsf{gMark}$ encodings of their schemas $\texttt{LSN}$, $\texttt{SP}$, $\texttt{WD}$ (as well as our user-defined schema $\texttt{Bib}$) to generate more diverse query loads, going beyond the existing ones w.r.t.\ the fine-grained control of different characteristics such as size, selectivity, and recursion.

\begin{table}[t]\centering
\caption{$\alpha(Q)$ averaged across constant, linear, and quadratic queries (with standard deviation), with varying graph sizes, data, and query diversity.}
\label{table:averages}
\vspace*{-0.3cm}\scriptsize
\begin{tabular}{|l|c|c|c|}\hline
   & \em Constant & \em Linear & \em Quadratic   \\
  \hline
  $\texttt{LSN}$\em-Len & 0.200$\pm$0.417 & 1.189$\pm$0.261 & 2.032$\pm$0.059 \\
  $\texttt{LSN}$\em-Dis & 0.182$\pm$0.364 & 1.325$\pm$0.318 & 2.046$\pm$0.074 \\
  $\texttt{LSN}$\em-Con & 0.190$\pm$0.391 & 1.244$\pm$0.326 & 2.017$\pm$0.032 \\
  $\texttt{LSN}$\em-Rec & 0.196$\pm$0.409 & 1.090$\pm$0.492 & 1.564$\pm$0.889 \\\hline
  $\texttt{Bib}$\em-Len & 0.003$\pm$0.010 & 0.921$\pm$0.122 & 1.405$\pm$0.337 \\
  $\texttt{Bib}$\em-Dis  & 0.000$\pm$0.000 & 0.995$\pm$0.012 & 1.607$\pm$0.261 \\
  $\texttt{Bib}$\em-Con  & 0.023$\pm$0.029  & 0.986$\pm$0.112 & 1.409$\pm$0.296 \\
  $\texttt{Bib}$\em-Rec  & 0.100$\pm$0.316 & 0.982$\pm$0.073 & 1.493$\pm$0.335 \\\hline
  $\texttt{WD}$\em-Len &  0.016$\pm$0.044 & 1.427$\pm$0.392 & 2.004$\pm$0.022\\
  $\texttt{WD}$\em-Dis  & 0.009$\pm$0.022 & 1.412$\pm$0.380 & 1.999$\pm$0.014\\
  $\texttt{WD}$\em-Con  & -0.010$\pm$0.026 & 1.540$\pm$0.495 & 1.750$\pm$0.708\\
  $\texttt{WD}$\em-Rec  & 0.587$\pm$0.830 & - & 1.976$\pm$0.012 \\\hline
  $\texttt{SP}$\em  &  0.074$\pm$0.130 & 1.064$\pm$0.034 & 2.034$\pm$0.295 \\\hline
  \end{tabular}
\end{table}

\begin{figure}[t]\centering
\includegraphics[width=5.7cm]{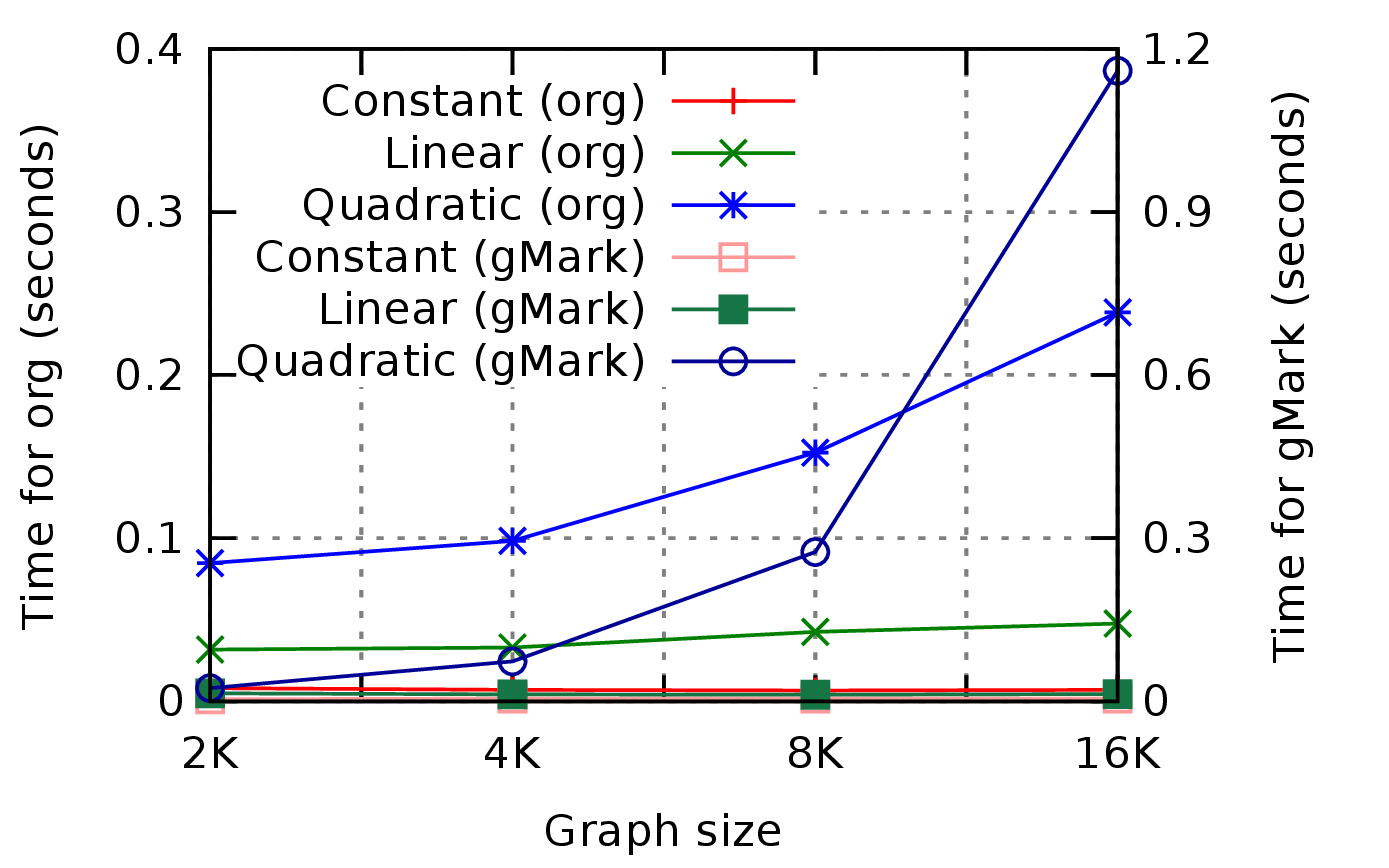}
\vspace*{-0.3cm}
 \caption{\label{fig:plots0} Evaluation times of a constant, a linear and a quadratic query in the original SP2Bench query load (org) and $\mathsf{gMark}$, respectively.}
\vspace*{-0.3cm}
\end{figure}

\begin{figure*}\centering
\subfigure[Est.\,selectivities:\,$\texttt{Bib}$\em -Len.]{
\includegraphics[width=0.2\textwidth]{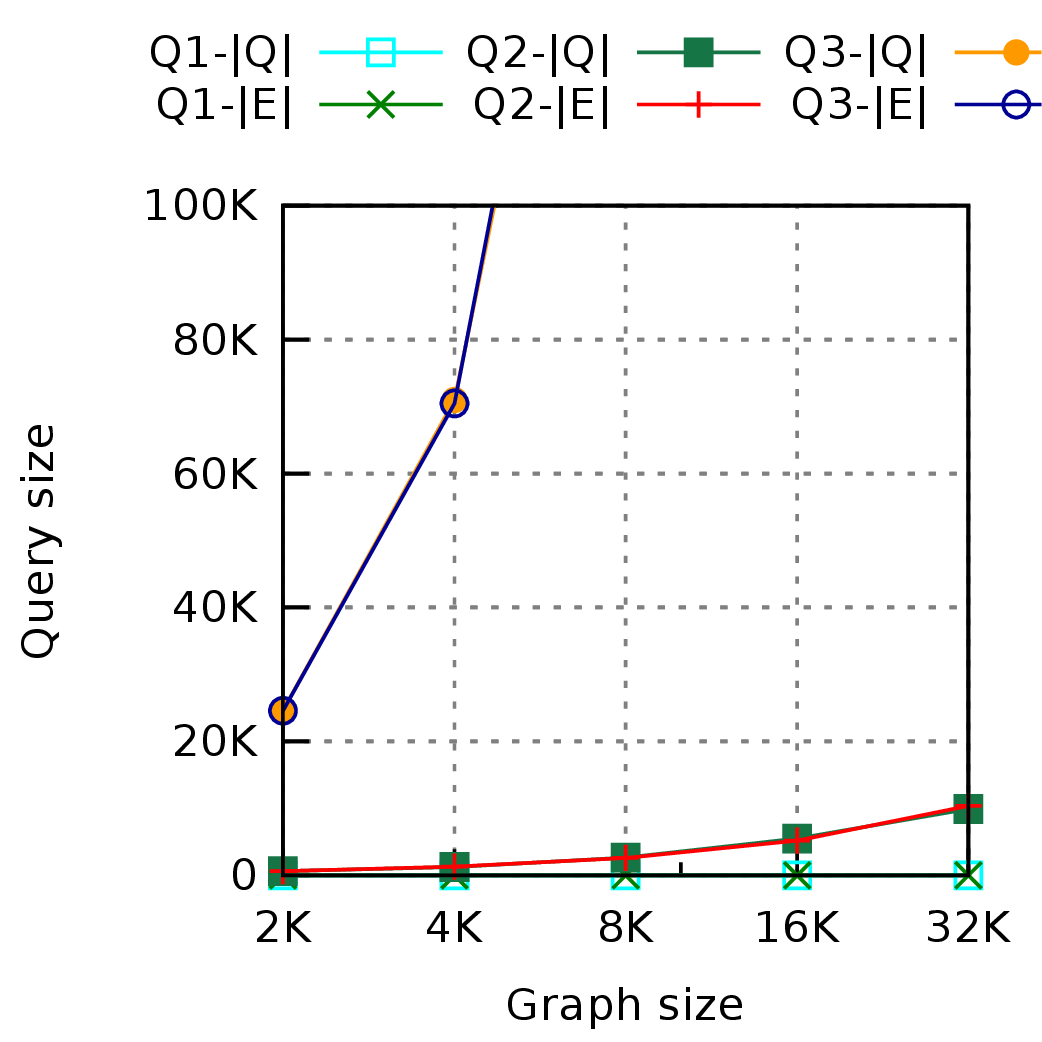}
} ~~~~~~
\subfigure[Est.\,selectivities:\,$\texttt{Bib}$\em -Con.]{
\includegraphics[width=0.2\textwidth]{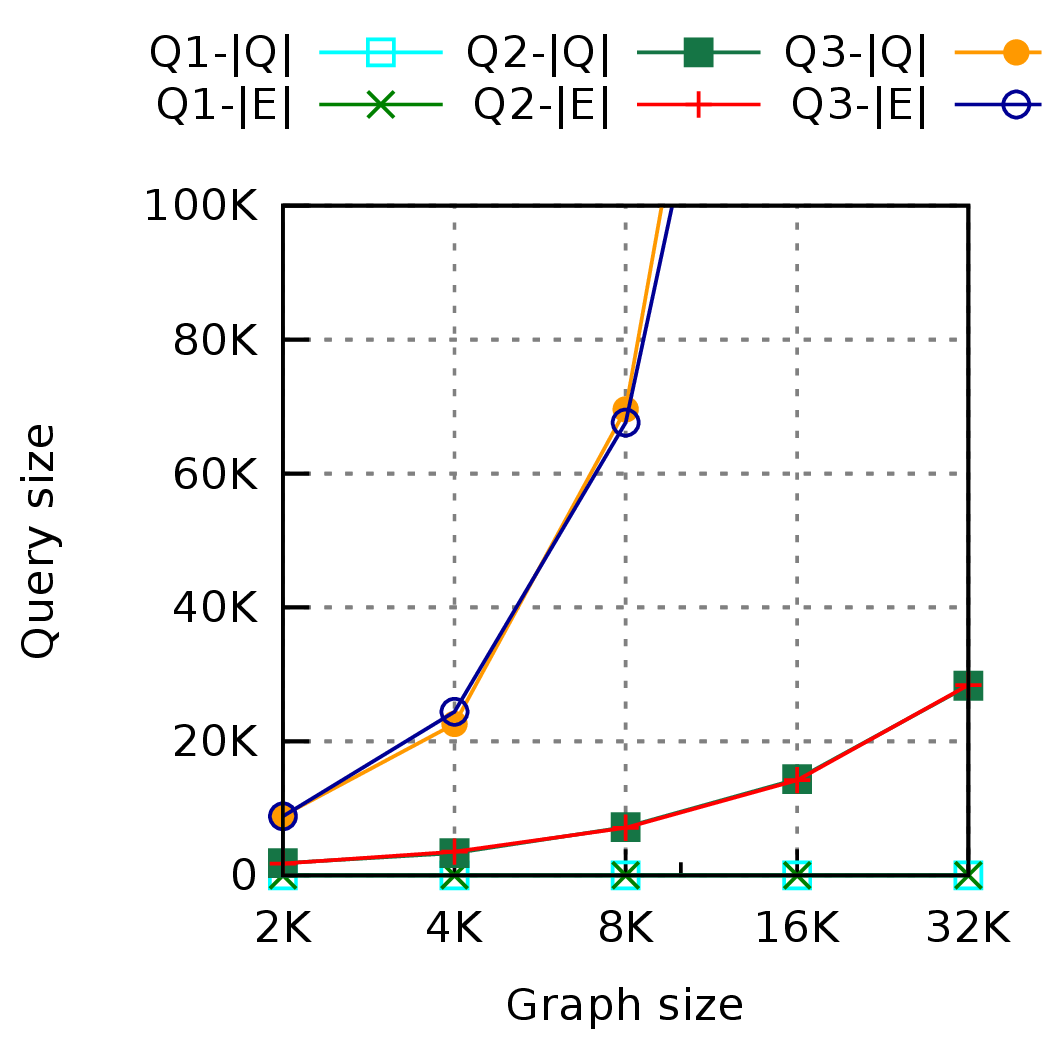}
}~~~~~~
\subfigure[Est.\,selectivities:\,$\texttt{Bib}$\em -Dis.]{
\includegraphics[width=0.2\textwidth]{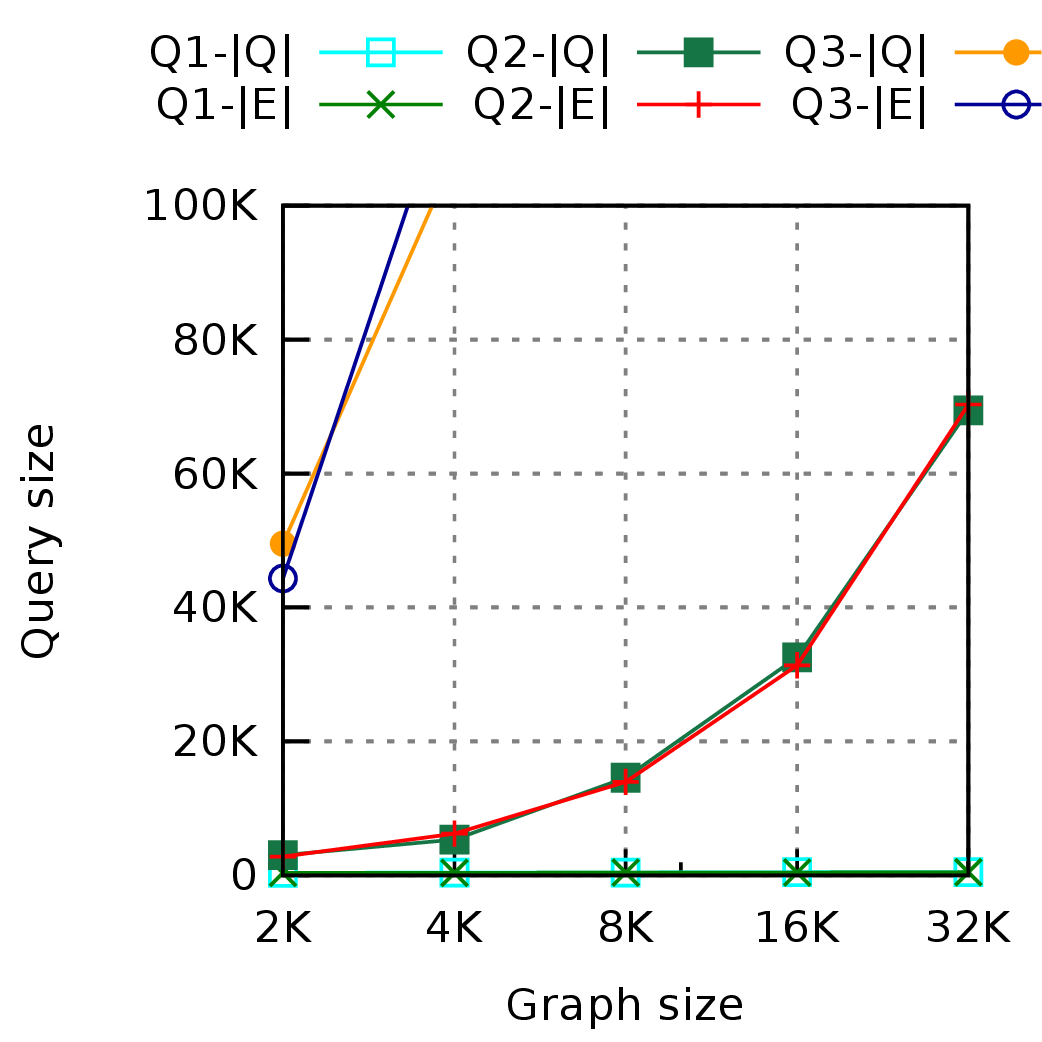}
}~~~~~~
\subfigure[Est.\,selectivities:\,$\texttt{Bib}$\em -Rec.]{
\includegraphics[width=0.2\textwidth]{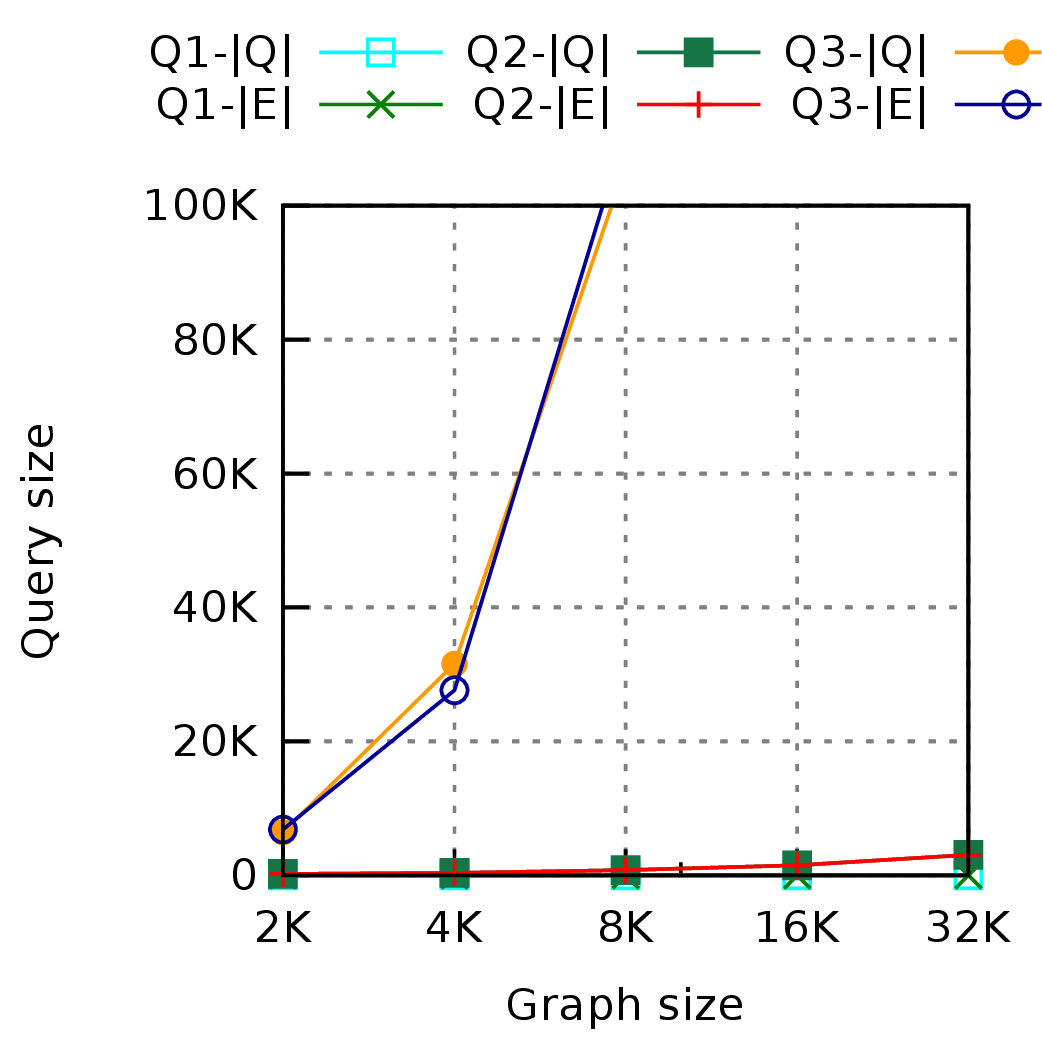}
}
\vspace*{-0.2cm}
\caption{\label{fig:plots1}Summary of selectivity estimation 
quality results for the bibliographical use case $\texttt{Bib}$.}
\vspace*{-0.3cm}
\end{figure*}

\subsection{Quality and scalability}
\label{subsect:scala}

{\smallskip\noindent\bf Selectivity estimation quality study.}
Our first set of experiments focuses on understanding the quality of the
selectivity estimation performed using the algebra presented in Section~\ref{sec:querygen}. 
We used four types of workloads in which we stress-test several diverse queries with incrementally varied query size:
{\em Len} generates queries with varying path lengths, no disjuncts, no conjuncts, and no recursion; 
{\em Dis} generates queries with disjuncts, no conjuncts and no recursion; 
{\em Con} generates queries with conjuncts and disjuncts and no recursion;
{\em Rec} generates queries with recursion (Kleene-stars). 
Each of the above tests
leads to a workload of $30$ queries, out of which $10$ are constant, $10$ are linear, and $10$ are quadratic. 
We repeated these experiments for each of the $\texttt{Bib}$, $\texttt{WD}$, and $\texttt{LSN}$ use cases (cf.\ Section~\ref{subsec:scenarios}).
We executed each query $Q$ on graph instances of sizes between $2$K and $32$K, 
and we counted the number of results returned on each instance. 
To compute the $\alpha$-value in the formula $|Q(G)| = \beta
|G|^\alpha$ (cf.\ Section~\ref{sec:querygen}), we computed a simple linear regression between $\log |G|$ and
$\log |Q(G)|$.  
We averaged the $\alpha$-values obtained across all queries belonging to the
same selectivity class.
We report these results, together with the standard deviation, in Table~\ref{table:averages}. 
As expected, we observe that, for all combinations of use cases and diverse query workloads,
the estimated values of $\alpha$ are $\approx\!0$ for constant queries,
$\approx\!1$ for linear queries, and $\approx\!2$ for quadratic queries.  The missing value in Table~\ref{table:averages} corresponds to linear recursive queries for $\texttt{WD}$. 
In the presence of recursion, we actually observed numerous failures on the majority of the studied systems (cf.\ Section~\ref{sec:databasesStudy}).

Finally, the last row ($\texttt{SP}$) of Table~\ref{table:averages} shows the estimated $\alpha$-values for SP2Bench~\cite{SchmidtHLP09} on a set of queries following our $\mathsf{gMark}$ encoding of the original set of SP2Bench queries. In the remainder of our study in Section~\ref{sec:databasesStudy}, we disregard 
$\texttt{SP}$ since this use case does not bring more insights than the query sets generated for the other use cases.

To further illustrate the precision of our estimated values, we report in
Fig.~\ref{fig:plots1} the estimated selectivities ($|E|$) along with the theoretical selectivities ($|Q|$) for constant ($Q_1$), linear ($Q_2$), and quadratic ($Q_3$) queries on the $\texttt{Bib}$ use case. 
We observe that for the classes of queries of increasing expressiveness the
number of results is generally higher for quadratic queries, while it is
linearly and constantly varying for the other queries, as expected.
We also observe that the two curves representing the estimated
selectivities and the theoretical ones closely overlap in all the cases. 
Finally, notice that the above experiments are considering chain queries only. The results on other shapes and/or use cases are similar and omitted for conciseness. 

We conclude from this study that the schema-driven $\mathsf{gMark}$ selectivity estimation framework generates consistently high quality estimates, across all selectivity classes, across a broad spectrum of diverse data sets and queries.

\begin{table}[t]
\caption{Graph generation time for varying graph sizes (\# nodes) and schemas.} \label{table:gmarktimes}
\vspace*{-0.3cm}\scriptsize\centering
\begin{tabular}{|c|c|c|c|c|}
  \hline
  & $100$K & $1$M & $10$M & $100$M \\
  \hline
  $\texttt{Bib}$ & 0m0.057s & 0m0.638s & 0m8.344s & 1m28.725s \\
  \hline
  $\texttt{LSN}$  & 0m0.225s & 0m1.451s & 0m23.018s & 3m11.318s \\
  \hline
  $\texttt{WD}$ & 0m2.163s & 0m25.032s & 4m10.988s & 113m31.078s\\ 
 \hline
 $\texttt{SP}$ & 0m0.638s	 & 0m7.048s & 1m28.831s & 15m23.542s \\
 \hline
\end{tabular}
\vspace*{-0.2cm}
\end{table}

{\smallskip\noindent\bf Scalability study.}
Our second set of experiments is devoted to measuring the time taken by the
graph generator of $\mathsf{gMark}$, while varying the size
of the data and the size of the query workload. 
To gauge the robustness of our graph instance generator w.r.t.\ data and
query diversity, in these experiments we employed all four use cases: $\texttt{Bib}$, $\texttt{WD}$, $\texttt{LSN}$, and $\texttt{SP}$.
We report the results in Table~\ref{table:gmarktimes}.
We observe that the generator scales quite well for all use cases. It is quite efficient for big graph sizes in all cases except $\texttt{WD}$. This is due to the quite complex nature of its schema, which induces much denser graph instances compared to the other use cases.  
For example, $\texttt{WD}$ instances have two orders of magnitude higher number of edges than $\texttt{Bib}$ instances having the same number of nodes.
This is not a limitation of $\mathsf{gMark}$ or of the $\texttt{WD}$ scenario, but rather a specific feature of very dense graphs.

We conclude by observing that $\mathsf{gMark}$ can efficiently generate both small and large graph instances, on a diversity of practical scenarios.
We note that we also conducted a scalability study of query generation, which showed that $\mathsf{gMark}$ easily generates workloads of a thousand queries for $\texttt{Bib}$, $\texttt{LSN}$, and $\texttt{SP}$ in around one second and for the richer $\texttt{WD}$ scenario in around 10 seconds.  Query translation of a thousand queries into all four supported syntaxes for each of the four scenarios took a mere tenth of a second. This study shows that $\mathsf{gMark}$ workload generation is very efficient and scalable for large-scale complex workloads.

\section{Evaluation of query engines}
\label{sec:databasesStudy}

We next turn to an empirical evaluation of a representative selection of
currently available graph query processing engines using $\mathsf{gMark}$.  Our goal here
is both to demonstrate the new capabilities in benchmarking introduced by
$\mathsf{gMark}$, and to pinpoint limitations and areas for further improvement in
current graph query processing solutions.

\subsection{Design of experiments}
\label{subsec:experimentaldesign}

{\smallskip\noindent\bf Systems.} 
The database systems (and their supported query languages) that we consider in our study are\footnote{For obvious reasons, we obfuscate the names of the three commercial systems employed in our study.}:
\smallskip

\noindent$\bullet$ $\textbf{G}$: a native graph database (openCypher~\cite{opencypher})

\noindent$\bullet$ $\textbf{S}$: a popular SPARQL query engine (SPARQL~1.1~\cite{sparql})

\noindent$\bullet$ $\textbf{P}$: PostgreSQL v9.3.9 (SQL:1999 recursive views\footnote{We use the standard translation of $\mathsf{UCRPQ}$'s into recursive views, implemented using linear recursion~\cite{Baeza13}.}~\cite{postgres}) 

\noindent$\bullet$ $\textbf{D}$: a modern Datalog engine (Datalog~\cite{Ullman89})
\smallskip

For the sake of fairness, we used default configurations for each system i.e., without special purpose optimizations.

{\smallskip\noindent\bf Query languages.} 
Recall that the queries generated by $\mathsf{gMark}$ are $\mathsf{UCRPQ}$'s. 
We provide in Appendix~B in our technical report~\cite{BBCFLA16} a translation of an
example $\mathsf{UCRPQ}$ into each of the above concrete syntaxes.
We note that not all systems support the full
expressive power of $\mathsf{UCRPQ}$'s.  In particular, arbitrary
$\mathsf{UCRPQ}$'s can be expressed in SPARQL, SQL, and Datalog, while openCypher supports
only those $\mathsf{UCRPQ}$'s having no occurrences of inverse or concatenation under Kleene
star. In our results regarding recursive queries, some of the  generated benchmark
queries do indeed exhibit inverse and/or concatenation in a
recursive conjunct.  In these cases, the corresponding openCypher query has only the non-inverse
symbol and/or the first symbol in a concatenation of symbols, respectively.  Furthermore,
while all other languages adopt the classical homomorphic semantics for
conjunctive queries~\cite{Ullman89}, openCypher adopts an isomorphic
semantics.
For these two reasons, openCypher queries often have
answer sets that differ from that of their counterparts in the other languages, which should be kept
in mind while evaluating experimental results pertaining to system $\textbf{G}$.

{\smallskip\noindent\bf Measurements.} 
We generate and execute query workloads on a variety of
graph configurations.  We execute and measure the runtime of each query six
times. The first one is a ``cold'' run that we exclude from the computation of
the average; from the remaining five ``warm'' runs we drop the fastest and
slowest and then report the average of the remaining three execution times.
Between the execution of each query, we close and reopen the database to
clear any caching effects.  We consider the following parameters in generating
workloads: selectivity classes of the tested queries (constant, linear,
quadratic) and size (amounting to $30$ queries for each workload).  We consider
the following parameters in generating graphs: size  (from $2$K to
$16$K nodes) and use case i.e, $\texttt{Bib}$, $\texttt{LSN}$, and $\texttt{WD}$ (cf.\ Section~\ref{subsec:scenarios}). 

We make three remarks here. 
(i) Despite being small, the considered graph sizes were already sufficient to
illustrate interesting behavior and distinctions between the studied systems.  Indeed, as we discuss in Section~\ref{sec:results:of:experiments}, already on these graphs we observed that many simple queries fail on a majority of the systems.
Even for those queries that succeed, the evaluation times are often very high e.g.,
hundreds or even thousands of seconds already on instances of these sizes.
(ii) To ensure a fair comparison of all systems and to avoid measuring the time to print the query results, we added to all queries the aggregate
$\mathit{count}(\mathit{distinct} (\overline{?v}) )$, where $\overline{?v}$ is the  (binary) vector of output variables. We recall that $\mathit{distinct}$ is also necessary for our
analysis since the algebra relies on the elimination of duplicates (cf.\ Section~\ref{sec:query-generation-selectivity}).
(iii) In this study, all queries are chains, as this is the basic query shape from which the others are constructed and hence sufficient for illustrating the relative performance of current graph query processing engines.
Finally,  even though
in the presentation of the results, we focus on the default use case
$\texttt{Bib}$, we observed comparable trends for the other use cases, that we omit for the sake of conciseness.

\begin{figure*}\centering
\hspace*{-0.3cm}
\subfigure[\label{fig:plots2:constant}{\em Constant} queries.]{
\includegraphics[width=0.3\textwidth]{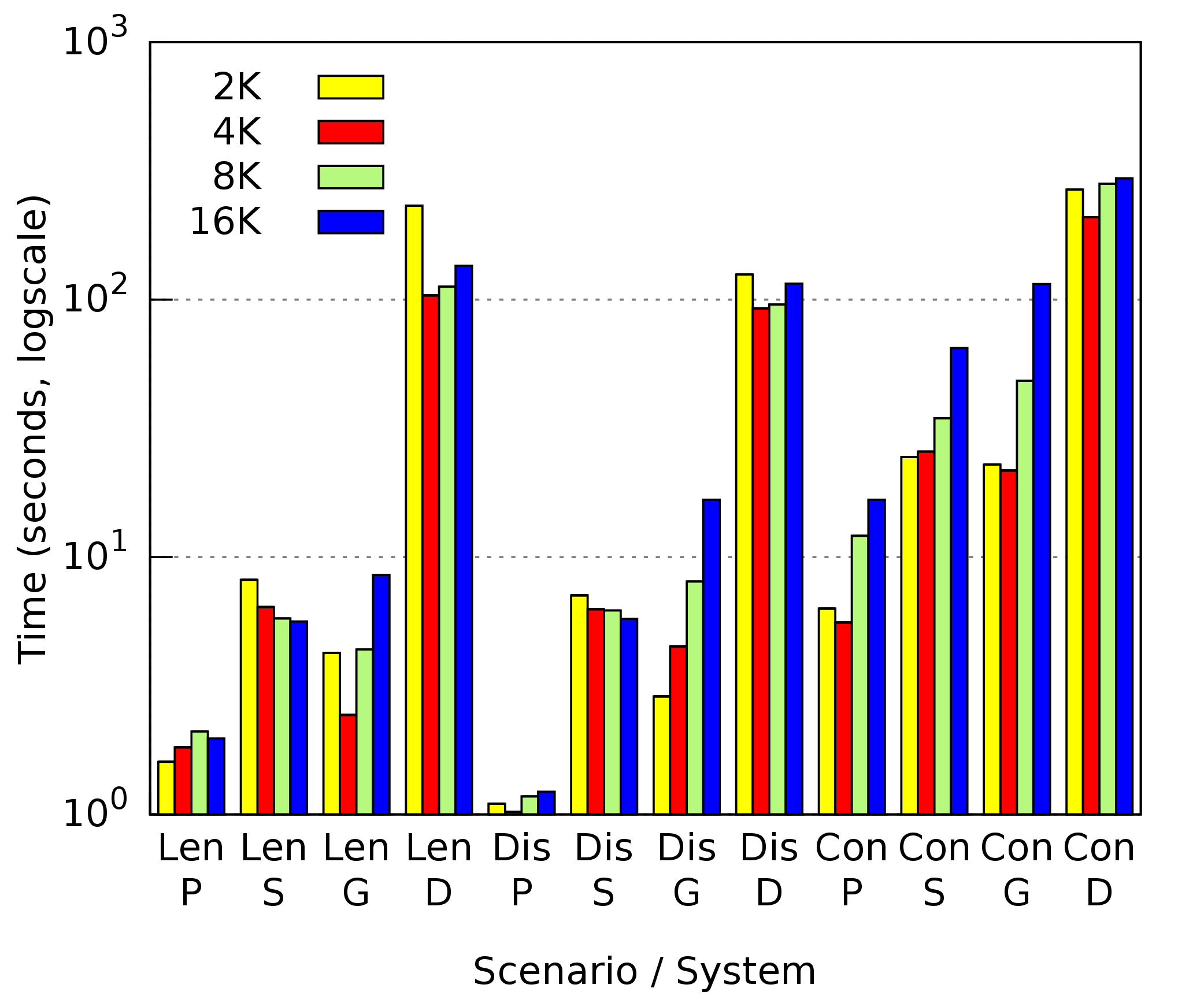}
}~~~~~~
\hspace*{-0.52cm}
\subfigure[\label{fig:plots2:linear}{\em Linear} queries.]{
\includegraphics[width=0.3\textwidth]{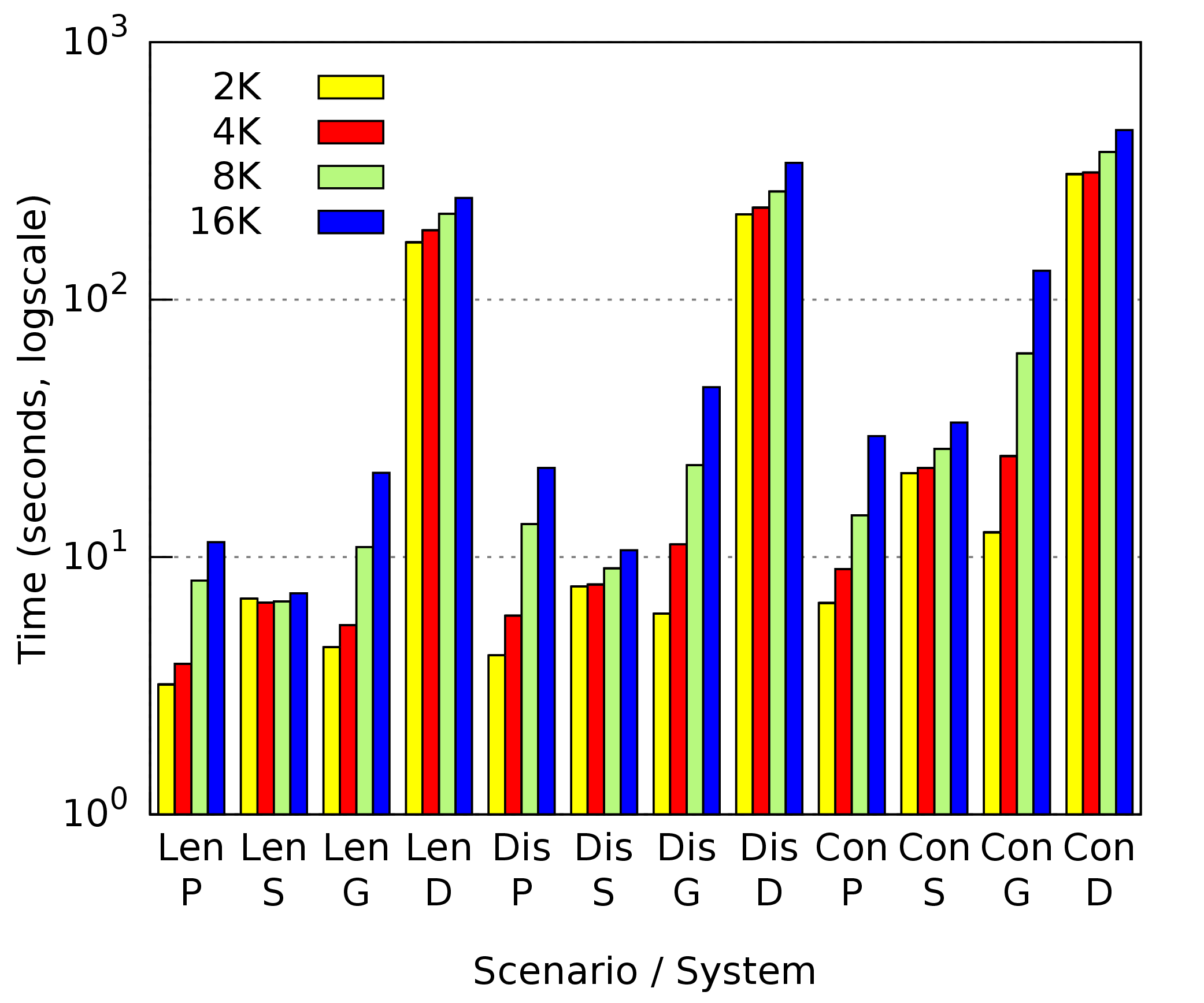}
}~~~~~~
\hspace*{-0.52cm}
\vspace*{-0.1cm}
\subfigure[\label{fig:plots2:quadratic}{\em Quadratic} queries.]{
\includegraphics[width=0.3\textwidth]{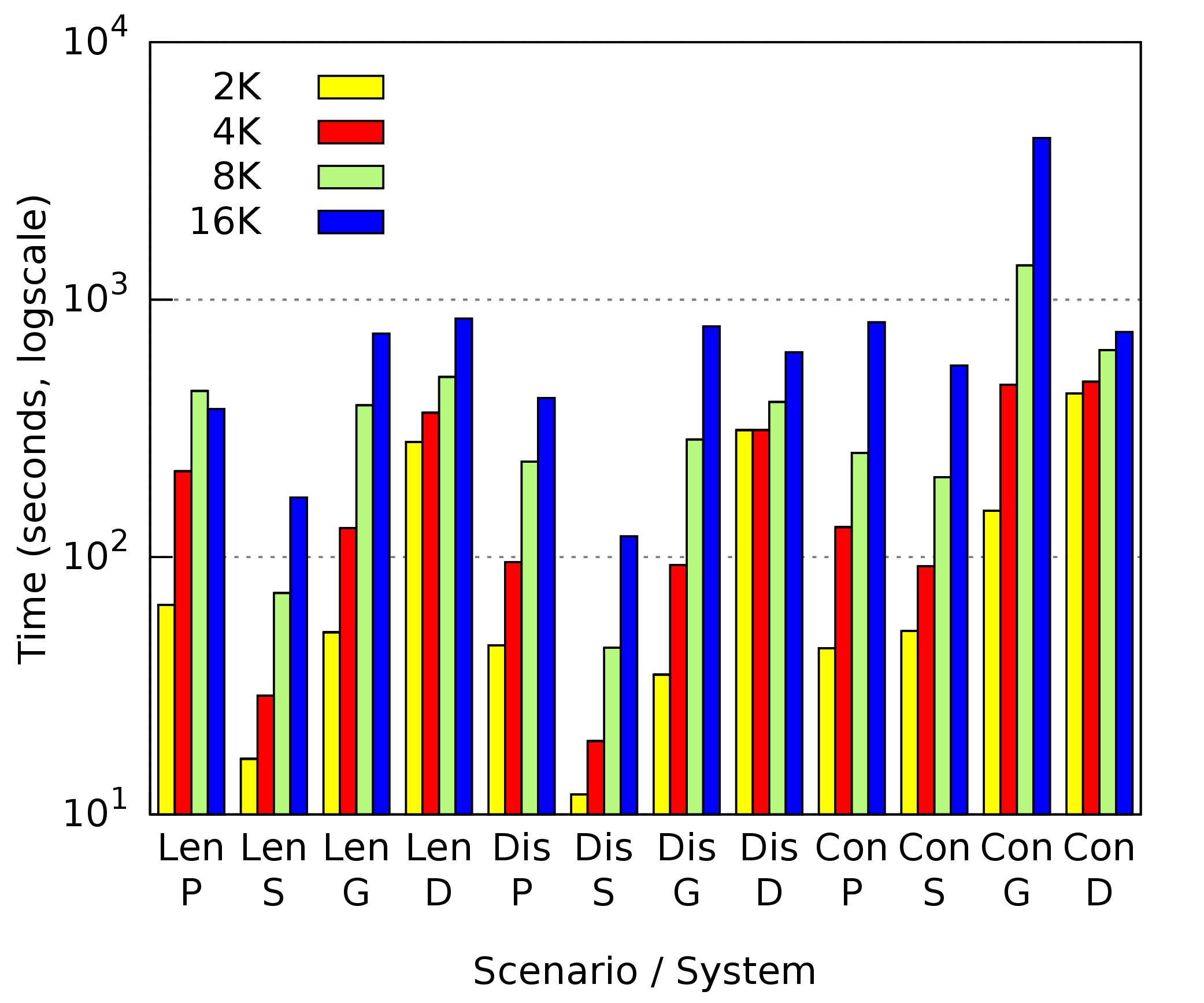}
}
\vspace*{-0.1cm}
\caption{\label{fig:plots2}Summary of query execution times for diverse query workloads ({\em Len}, {\em Dis}, {\em Con})
and various graph sizes under PostgreSQL ($\textbf{P}$) and three commercial systems: a SPARQL query engine ($\textbf{S}$), a native graph database ($\textbf{G}$), and a modern Datalog engine ($\textbf{D}$).}
\vspace*{-0.1cm}
\end{figure*}

\subsection{Results of experiments}\label{sec:results:of:experiments}
{\smallskip\noindent\bf Non-recursive queries.}
In our first experiment, we focus on the non-recursive queries i.e., 
query workloads {\em Len}, {\em Dis}, {\em Con} (cf.\ Section~\ref{subsect:scala}).
We summarize the results in Fig.~\ref{fig:plots2}.
The goal of this study is to observe how the different systems react to
these varied workloads. Fig.~\ref{fig:plots2:constant} shows the query execution
times averaged across the $10$ constant queries of each workload. 
Out of the $10$ averages obtained by the $5$ warm runs of each query, we
computed once again the overall average, by discarding two out of the $10$ averages
that have the farthest standard deviation with respect to this overall
average. This allows to capture the cases in which some of the systems fail
or give outlier results. 
In Fig.~\ref{fig:plots2:constant}, we observe that $\textbf{P}$ reacts better than $\textbf{S}$, $\textbf{G}$, and $\textbf{D}$ to query workload diversity, by exhibiting lower query evaluation times on all instance sizes. This behavior is confirmed in the case of linear queries, as shown in
Fig.~\ref{fig:plots2:linear}, for the {\em Con}
query workload on all the sizes and for the {\em Len} and {\em Dis} query workloads
for smaller sizes only i.e., $2$K and $4$K.
For larger instance sizes i.e., $8$K and $16$K, the behavior is reverted in favor of $\textbf{S}$ for linear queries. 
Then, as shown in Fig.~\ref{fig:plots2:quadratic} for quadratic queries, $\textbf{S}$ continues on this trend by beating $\textbf{P}$, $\textbf{G}$, and $\textbf{D}$.  
We also observe from all query execution times reported in Fig.~\ref{fig:plots2} that the times taken by constant and linear are
of the same order of magnitude, whereas quadratic queries, as expected, typically
exhibit an order of magnitude slowdown. There is only one system ($\textbf{D}$) for
which the differences of the behavior along the sets of linear and quadratic  
queries are blurred. 

We conclude from this study that $\mathsf{gMark}$ allows us to generate interesting
queries and diverse query workloads that already on small graph instances
stress-test state-of-the-art systems, and highlight particular strengths and
weaknesses in processing graph queries.  
As a general observation, we can further conclude that  
the straightforward standard 
implementation of  $\mathsf{UCRPQ}$'s in
PostgreSQL typically shows superior performance across a broad class of queries (i.e.
constant and linear) to that of existing dedicated systems.

{\smallskip\noindent\bf Recursive queries.}
Our second experiment is devoted to recursive queries, generated
by query workloads containing Kleene stars.
Unfortunately, all tested systems either failed on the majority of these
queries or had to be manually terminated after unexpectedly long running times.  
For these reasons, it is difficult to draw a clear conclusion on recursive
queries.
Therefore, we performed a small case analysis: we considered two recursive queries of constant and quadratic selectivity, respectively, for which we
could collect results for at least one of the four systems.
We report the results for both queries in Table~\ref{tbl:recursive:queries}.
The first query has constant selectivity. 
$\textbf{P}$ was quite slow at evaluating it on small instances and starts
failing on graphs of $8$K nodes. 
$\textbf{S}$ was able to answer this query only on the smallest graph size ($2$K). 
$\textbf{G}$ failed in all cases and always returned empty results (due to its different query semantics, as discussed in Section~\ref{subsec:experimentaldesign}). 
The only system for which we could measure the evaluation time for all sizes was $\textbf{D}$, which also turned to be the most efficient one.
The second query has quadratic selectivity and only $\textbf{D}$ was able to evaluate it.
We conclude from this study that only $\textbf{D}$ is currently able to deal with
recursive queries.

\begin{table}
\caption{\label{tbl:recursive:queries}Execution time (sec.) for recursive queries.}
\vspace*{-0.2cm}\scriptsize\centering
\begin{tabular}{|c|cccc|cccc|}\hline
\multirow{2}{*}{\hspace*{-0.1cm}\em Syst.\hspace*{-0.1cm}}
& \multicolumn{4}{c|}{\em Query 1: Graph Size} 
& \multicolumn{4}{c|}{\em Query 2: Graph Size} \\\hhline{~--------}
& $2$K & $4$K & $8$K & $16$K & $2$K & $4$K & $8$K & $16$K \\\hline 
$\textbf{P}$ & 3400 & \hspace*{-0.1cm}72113\hspace*{-0.15cm} & - & - & - & - & - & - \\
$\textbf{G}$ & - & - & - & - & - & - & - & - \\
$\textbf{S}$ & 6621 & - & - & - & - & - & - & - \\
$\textbf{D}$ & 450 & 455 & 552 & 725 & 607 & 704 & 1295 & 2095\\\hline
\end{tabular}
\vspace*{-0.1cm}
\end{table}

\section{Concluding remarks}
\label{sec:concl}
We presented $\mathsf{gMark}$, the first generator that satisfies the key criteria of being domain-independent, extensible,
schema-driven, and highly configurable also in terms of the expected query
selectivity of a given workload.  The latter is a novel contribution on its own
and is applicable to other independent benchmarks and problems.
For instance, we could envision the query workload generation in $\mathsf{gMark}$ applied to real graph data sets on top of which a schema extraction tool has been run beforehand.
Furthermore, $\mathsf{gMark}$ is the first benchmark to generate workloads
exhibiting recursive path queries, which are central to graph querying.
Our in-depth empirical study
demonstrated both the quality and practicality of $\mathsf{gMark}$. Moreover, our
experiments highlighted important limitations in the query processing
capabilities of current state-of-the-art graph processing engines, already
on small graph instances and on both recursive and non-recursive queries.

We plan to align our work on $\mathsf{gMark}$ with international benchmarking bodies such as LDBC~\cite{Erling2015}.
Looking ahead to the future work, there are many directions for
further investigation e.g., extending the selectivity estimation to $n$-ary queries.
We also aim to evangelize for the use of $\mathsf{gMark}$ by researchers in the graph data management community; a first step in this direction is our VLDB demo~\cite{BBCFLA16demo}.

\bibliographystyle{IEEEtran}
\bibliography{paper}

\vspace*{-1.2cm}
\begin{IEEEbiographynophoto}{Guillaume Bagan} (PhD, Université de Caen,
2009) is a CNRS research engineer at Université Lyon 1.
His research focuses on query evaluation for
relational and graph databases. 
\end{IEEEbiographynophoto}

\vspace*{-1.2cm}
\begin{IEEEbiographynophoto}{Angela Bonifati} (PhD, Politecnico di Milano,
2002) is a professor at Université Lyon 1, conducting research on
graph databases, data integration for complex data formats, and 
query inference.  
\end{IEEEbiographynophoto}

\vspace*{-1.2cm}
\begin{IEEEbiographynophoto}{Radu Ciucanu} (PhD, Université Lille 1 and Inria, 2015) is an associate professor at Université Blaise Pascal, Clermont-Ferrand, working mainly on graph databases and data integration.
\end{IEEEbiographynophoto}

\vspace*{-1.2cm}
\begin{IEEEbiographynophoto}{George Fletcher} (PhD, Indiana University Bloomington, 2007) is an associate professor at TU Eindhoven. He studies data-intensive systems, with a focus on the theory and engineering of query languages.

\end{IEEEbiographynophoto}

\vspace*{-1.2cm}
\begin{IEEEbiographynophoto}{Aurélien Lemay} (PhD, Université Lille 3, 2002) is an associate professor at Université Lille 3 and Inria. 
His research focuses on language theory, database theory, and grammatical inference. 
\end{IEEEbiographynophoto}

\vspace*{-1.2cm}
\begin{IEEEbiographynophoto}{Nicky Advokaat}
obtained a MSc in Computer Science in 2015 at TU Eindhoven, with a thesis on graph database benchmarking.
\end{IEEEbiographynophoto}

\end{document}